\documentclass[11pt]{article}
\usepackage{fullpage}
\usepackage{times}
\usepackage{subfigure}
\usepackage{color}
\usepackage{url}
\usepackage{graphicx}
\usepackage{amsmath}
\usepackage{amssymb}

\newcommand{\qed}{\mbox{}\hspace*{\fill}\nolinebreak\mbox{$\rule{0.6em}{0.6em}$}
}
\newcommand{\expect}{{\bf \mbox{\bf E}}}
\newcommand{\prob}{{\bf \mbox{\bf Pr}}}

\definecolor{gray}{rgb}{0.5,0.5,0.5}
\newcommand{\e}{{\epsilon}}

\newtheorem{theorem}{Theorem}[section]
\newtheorem{lemma}[theorem]{Lemma}

\newtheorem{corollary}[theorem]{Corollary}
\newtheorem{definition}[theorem]{Definition}
\newtheorem{remark}[theorem]{Remark}

\newenvironment{proof}{{\bf Proof:}}{$\qed$\par}
\newenvironment{proofnoqed}{{\bf Proof:}}{}
\newenvironment{proofof}[1]{\noindent{\bf Proof of #1:}}{$\qed$\par}

\usepackage{hyperref}
\hypersetup{
bookmarksnumbered
}

\newcommand{\coll}{{\mathcal R}}
\newcommand{\colltwo}{{\mathcal T}}
\newcommand{\collpart}{{\mathcal S}}
\newcommand{\edgesets}{\mathcal X}

\begin{document}

\title{\Large Perfect Matchings in $\tilde{O}(n^{1.5})$ Time in Regular Bipartite Graphs}
\author{Ashish Goel\thanks{
    Departments of Management Science and Engineering and (by courtesy)
    Computer Science, Stanford University.
    Email: {\tt ashishg@stanford.edu}.
    Research supported by NSF
    ITR grant 0428868, NSF CAREER award 0339262, and a grant from the
    Stanford-KAUST alliance for academic excellence.}\\
\and Michael Kapralov\thanks{
    Institute for Computational and Mathematical Engineering, Stanford University.
    Email: {\tt kapralov@stanford.edu}. Research supported by a Stanford Graduate Fellowship.}\\   
\and Sanjeev Khanna\thanks{Department of Computer and Information Science, University of Pennsylvania,
Philadelphia PA. Email: {\tt sanjeev@cis.upenn.edu}. Supported in
part by a Guggenheim Fellowship, an IBM Faculty Award, and by NSF Award CCF-0635084.}
}

\date{}
\maketitle
\setcounter{page}{0}
\thispagestyle{empty}
\pdfbookmark[1]{Title and abstract}{Myabstract}
\begin{abstract}

  We consider the well-studied problem of finding a perfect matching in
  $d$-regular bipartite graphs with $2n$ vertices and $m = nd$ edges. While
  the best-known algorithm for general bipartite graphs (due to Hopcroft and
  Karp) takes $O(m \sqrt{n})$ time, in regular bipartite graphs, a perfect
  matching is known to be computable in $O(m)$ time.  Very recently, the
  $O(m)$ bound was improved to $O(\min\{m, \frac{n^{2.5}\ln n}{d}\})$ expected
  time, an expression that is bounded by $\tilde{O}(n^{1.75})$.  In this
  paper, we further improve this result by giving an $O(\min\{m,
  \frac{n^2\ln^3 n}{d}\})$ expected time algorithm for finding a perfect
  matching in regular bipartite graphs; as a function of $n$ alone, the
  algorithm takes expected time $O((n\ln n)^{1.5})$.

  To obtain this result, we design and analyze a two-stage sampling scheme that reduces
  the problem of finding a perfect matching in a regular bipartite graph to the
  same problem on a subsampled bipartite graph with $O(n\ln n)$ edges. The
  first-stage is a sub-linear time uniform sampling that reduces the size of
  the input graph while maintaining certain structural properties of the
  original graph. The second-stage is a non-uniform sampling that takes
  linear-time (on the reduced graph) and outputs a graph with $O(n\ln n)$
  edges, while preserving a matching with high probability.  This matching is
  then recovered using the Hopcroft-Karp algorithm. While the standard
  analysis of Hopcroft-Karp also gives us an $\tilde{O}(n^{1.5})$ running
  time, we present a tighter analysis for our special case that results in the
  stronger $\tilde{O}(\min\{m, \frac{n^2}{d} \})$ time mentioned earlier.

  Our proof of correctness of this sampling scheme uses
  a new correspondence theorem between cuts and Hall's theorem ``witnesses'' for a perfect matching in a bipartite
  graph that we prove. We believe this theorem may be of independent
  interest; as another example application, we show that a perfect matching in
  the support of an $n \times n$ doubly stochastic matrix with $m$ non-zero
  entries can be found in expected time $\tilde{O}(m + n^{1.5})$.
\end{abstract}
\newpage

\section{Introduction}

A bipartite graph $G = (P, Q, E)$ with vertex set $P \cup Q$ and edge set $E
\subseteq P \times Q$ is said to be regular if every vertex has the same
degree $d$. We use $m=nd$ to denote the number of edges in $G$ and $n$ to
represent the number of vertices in $P$ (as a consequence of regularity, $P$
and $Q$ have the same size). Regular bipartite graphs are a fundamental
combinatorial object, and arise, among other things, in expander
constructions, scheduling, routing in switch fabrics, and
task-assignment~\cite{mr:random,amsz:color2003,cos:regular2001}.

A regular bipartite graph of degree $d$ can be decomposed into exactly $d$
perfect matchings, a fact that is an easy consequence of Hall's
theorem~\cite{b:graphtheory}, and is closely related to the Birkhoff-von
Neumann decomposition of a doubly stochastic matrix~\cite{b:bvn46,vn:bvn53}.
Finding a matching in a regular bipartite graph is a well-studied problem,
starting with the algorithm of K\"{o}nig in 1916~\cite{k:regular16}, which is
now known to run in time $O(mn)$. The well-known bipartite matching algorithm
of Hopcroft and Karp~\cite{hk:match73} can be used to obtain a running time of
$O(m\sqrt n)$. In graphs where $d$ is a power of 2, the following elegant
idea, due to Gabow and Kariv~\cite{gk:edge1982}, leads to an algorithm with
$O(m)$ running time. First, compute an Euler tour of the graph (in time
$O(m)$) and then follow this tour in an arbitrary direction. Exactly half the
edges will go from left to right; these form a regular bipartite graph of
degree $d/2$. The total running time $T(m)$ thus follows the recurrence $T(m)
= O(m) + T(m/2)$ which yields $T(m) = O(m)$. Extending this idea to the
general case proved quite hard, and after a series of improvements (eg. by
Cole and Hopcroft~\cite{ch:color82}, and then by Schrijver~\cite{s:color99} to
$O(md)$), Cole, Ost, and Schirra~\cite{cos:regular2001} gave an $O(m)$
algorithm for the case of general $d$. Their main interest was in edge
coloring of general bipartite graphs, where finding perfect matchings in
regular bipartite graphs is an important
subroutine.
Very recently, Goel, Kapralov, and Khanna~\cite{gkk:rbp08}, gave a
sampling-based algorithm that computes a perfect matching in $d$-regular
bipartite graphs in $O(\min\{m, \frac{n^{2.5}\ln n}{d}\})$ expected time, an
expression that is bounded by $\tilde{O}(n^{1.75})$.  The algorithm
of~\cite{gkk:rbp08} uses uniform sampling to reduce the number of edges in the
input graph while preserving a perfect matching, and then runs the
Hopcroft-Karp algorithm on the sampled graph.

\paragraph{Our Results and Techniques:}
We present a significantly faster algorithm for finding perfect matchings in
regular bipartite graphs.
\begin{theorem}
There is an $O\left(\min\{m, \frac{n^2\ln^3 n}{d}\}\right)$ expected time
algorithm to find a perfect matching in a $d$-regular bipartite graph $G$.
\end{theorem}
As a function of $n$ alone, the running time stated above is $O((n\ln
n)^{1.5})$. Since the $O(m)$ running time is guaranteed by the algorithm of
Cole, Ost, and Schirra, we are only concerned with the case where $d$ is
$\Omega(\sqrt{n} \ln n)$. For this regime, our algorithm reduces the perfect
matching problem on a regular bipartite graph $G$ to the same problem on a
(not necessarily regular) sparse bipartite graph $H$ with $O(n\ln n)$
edges. This reduction takes time $O(\frac{n^2\ln^3 n}{d})$. We then use the
Hopcroft-Karp algorithm on $H$ to recover a perfect matching. A black-box use
of the analysis of the Hopcroft-Karp algorithm would suggest a running time of
$O(\frac{n^2\ln^3 n}{d} + n^{1.5}\ln n)$. However, we show that the final
sampled graph has some special structure that guarantees that the
Hopcroft-Karp algorithm would complete in time $O(\frac{n^2\ln^2 n}{d})$ whp.

For every pair $A\subseteq P, B\subseteq Q$, we define a {\em witness set}
$W(A,B)$ to be the set of all edges going from $A$ to $Q\setminus B$. Of
particular interest are what we call {\em Hall witness sets}, which correspond
to $|A| > |B|$; the well-known Hall's theorem~\cite{b:graphtheory} says that a
bipartite graph $H(P,Q,E_H)$ contains a perfect matching iff $E_H$ includes an
edge from each Hall witness set.
Thus any approach that reduces the size of the input bipartite graph by
sampling must ensure that some edge from every Hall witness set is included in
the sampled graph; otherwise the sampled graph no longer contains a perfect
matching.  Goel, Kapralov, and Khanna~\cite{gkk:rbp08} showed that no {\em
  uniform sampling} scheme on a $d$-regular bipartite graph can reduce the
number of edges to $o(\frac{n^2}{d \ln n})$ while preserving a perfect
matching, and hence their $\tilde{O}(n^{1.75})$-time algorithm is the best
possible running time achievable via uniform sampling followed by a black-box
invocation of the Hopcroft-Karp analysis.

In order to get past this barrier, we use here a two-stage sampling
process. The first stage is a uniform sampling (along the lines
of~\cite{gkk:rbp08}) which generates a reduced-size graph $G'=(P,Q,E')$ that
preserves not only a perfect matching but also a key relationship between the
sizes of ``relevant'' witness sets and cuts in the graph $G$. The second stage
is to run the non-uniform Bencz\'{u}r-Karger sampling scheme
~\cite{benczurkarger96} on $G'$ to generate a graph $G''$ with $\tilde{O}(n)$
edges while preserving a perfect matching w.h.p. Since this step requires
$\tilde\Omega(|E'|)$ time, we crucially rely on the fact that $G'$ does not
contain too many edges.

While our algorithm is easy to state and understand, the proof of correctness
is quite involved. The Bencz\'{u}r-Karger sampling was developed to generate, for
any graph, a weighted subgraph with $\tilde{O}(n)$ edges that approximately
preserves the size of all cuts in the original graph. The central idea
underlying our result is to show that there exists a collection of {\em core}
witness sets that can be identified in an almost one-one manner with cuts in
the graph such that the probability mass of edges in each witness set is
comparable to the probability mass of the edges in the cut identified with
it. Further, every witness set in the graph has a ``representative'' in this
collection of core witness sets. Informally, this allows us to employ
cut-preserving sampling schemes such as Bencz\'{u}r-Karger as
``witness-preserving'' schemes. We note here that the natural mapping which
assigns the witness set of a pair $(A,B)$ to the cut edges associated with
this pair can map arbitrarily many witness sets to the same cut and is not
useful for our purposes. One of our contributions is an uncrossing
theorem for witness sets, that we refer to as the {\em proportionate
  uncrossing theorem}.  Informally speaking, it says that given any collection
of witness sets $\coll$ such that the probability mass of each witness set is
comparable to that of its associated cut, there exists another collection
$\colltwo$ of witness sets such that (i) the natural mapping to cuts as
defined above is {\em half-injective} for $\colltwo$, that is, at most two
witness sets in $\colltwo$ map to any given cut, (ii) the probability mass of
each witness set is comparable to the probability mass of its associated cut,
and (iii) any subset of edges that hits every witness set in $\colltwo$ also
hits every witness set in $\coll$.  The collection $\colltwo$ is referred to
as a proportional uncrossing of $\coll$. As shown in
Figure~\ref{fig:cross}(a), we can not achieve an injective mapping, and hence
the half-injectivity is unavoidable.

We believe the half-injective correspondence between witness sets and cuts, as
facilitated by the proportionate uncrossing theorem, is of independent
interest, and will perhaps have other applications in this space of
problems. We also emphasize here that the uncrossing theorem holds for all
bipartite graphs, and not only regular bipartite graphs. Indeed, the graph
$G'$ on which we invoke this theorem does not inherit the regularity property
of the original graph $G$.  As another illustrative example, consider the
celebrated Birkhoff-von Neumann theorem~\cite{b:graphtheory,vn:bvn53} which
says that every doubly stochastic matrix can be expressed as a convex
combination of permutation matrices (i.e., perfect matchings).  In some
applications, it is of interest to do an iterative decomposition whereby a
single matching is recovered in each iteration. The best-known bound for this
problem, to our knowledge, is an $O(mb)$ time algorithm that follows from the
work of Gabow and Kariv~\cite{gk:edge1982}; here $b$ denotes the maximum
number of bits needed to express any entry in $M$. The following theorem is an
easy consequence of our proportionate uncrossing result.

\begin{theorem} \label{thm:bvn}
  Given an $n \times n$ doubly-stochastic matrix $M$ with $m$ non-zero
  entries, one can find a perfect matching in the support of $M$ in
  $\tilde{O}( m + n^{1.5})$ expected time.
\end{theorem}

The proof of this theorem and a discussion of known results about this problem are given in section ~\ref{sec:bvn}. Though this
result itself represents only a modest improvement over the earlier $O(mb)$
running time, it is an instructive illustration of the utility of the
proportionate uncrossing theorem.

It is worth noting that while the analysis of Goel, Kapralov, and Khanna was
along broadly similar lines (sample edges from the original graph, followed by
running the Hopcroft-Karp algorithm), the proportionate uncrossing theorem
developed in this paper requires significant new ideas and is crucial to
incorporating the non-uniform sampling stage into our algorithm. Further, the
running time of the Hopcroft-Karp algorithm is easily seen to be
$\Omega(m\sqrt{n})$ even for the 2-regular graph consisting of
$\Theta(\sqrt{n})$ disjoint cycles of lengths $2, 4, \ldots, \sqrt{n}$
respectively; the stronger analysis for our special case requires both our
uncrossing theorem as well as a stronger decomposition\footnote{It is known
  that the Hopcroft-Karp algorithm terminates quickly on bipartite
  expanders~\cite{motwani}, but those techniques don't help in our setting
  since we start with an arbitrary regular bipartite graph.}. As a step in
this analysis, we prove the independently interesting fact that after sampling
edges from a $d$-regular bipartite graph with rate $c \ln n\over d$, for some
suitable constant $c$, we obtain a graph that has a matching of size
$n-O(n/d)$ whp and such a matching can be found in $O(n/d)$ augmenting phases
of the Hopcroft-Karp algorithm whp.

\noindent
\paragraph{Organization:}
Section~\ref{sec:prelim} reviews and presents some useful corollaries of
relevant earlier work. In Section~\ref{sec:proportion}, we establish the
proportionate uncrossing theorem. In section~\ref{sec:algo}, we
present and analyze our two-stage sampling scheme, and
section~\ref{sec:improved-runtime} outlines the stronger analysis of the
Hopcroft-Karp algorithm for our special case. Section~\ref{sec:bvn} contains the proof of Theorem \ref{thm:bvn}
and a discussion of known results on finding perfect matchings in the support of double stochastic matrices.

\section{Preliminaries}
\label{sec:prelim}

In this section, we adapt and present recent results of Goel, Kapralov, and
Khanna~\cite{gkk:rbp08} as well as the Bencz\'{u}r-Karger sampling
theorem~\cite{benczurkarger96} for our purposes, and also prove a simple
technical lemma for later use.
\subsection{Bipartite Decompositions and Relevant Witness Pairs}
\label{sec:define}
Let $G=(P,Q,E)$ be a regular bipartite graph, with vertex set $P\cup Q$ and
edge set $E \subseteq P \times Q$. Consider any partition of $P$ into $k$ sets
$P_1, P_2, \ldots, P_k$, and a partition of $Q$ into $Q_1, Q_2, \ldots,
Q_k$. Let $G_i$ denote the (not necessarily regular) bipartite graph $(P_i,
Q_i, E_i)$ where $E_i = E \cap (P_i \times Q_i)$. We will call this a
``decomposition'' of $G$. 

Given $A\subseteq P$ and $B\subseteq Q$, define the witness set corresponding
to the pair $(A,B)$, denoted $W(A,B)$, as the set of all edges between $A$ and
$Q\setminus B$, and define the cut $C(A,B)$ as the set of all edges between $A
\cup B$ and $(P\setminus A) \cup (Q\setminus B)$. The rest of the definitions
in this section are with respect to some arbitrary but fixed decomposition of
$G$.

\begin{definition}
  An edge $(u,v)\in E$ is relevant if $(u,v)\in E_i$ for some $i$.
\end{definition}

\begin{definition}
  Let $E_R$ be the set of all relevant edges. A pair $(A,B)$ is said to be
  relevant if
  \begin{enumerate}
  \item $A\subseteq P_i$ and $B\subseteq Q_i$ for some $i$,
  \item $|A| > |B|$, and
  \item There does not exist another $A'\in P_i$, $B'\in Q_i$, such that
    $A'\subset A$, $|A'| > |B'|$, and $W(A',B') \cap E_R \subseteq W(A,B) \cap
    E_R$.
  \end{enumerate}
\end{definition}

Informally, a relevant pair is one which is contained completely within a
single piece in the decomposition, and is ``minimal'' with respect to that
piece. The following lemma is implicit in~\cite{gkk:rbp08} and is proved in
appendix~\ref{append:relevant-only} for completeness.
\begin{lemma}
\label{lem:relevant-only}
  Let $\coll$ denote all relevant pairs $(A,B)$ with respect to a
  decomposition of $G(P,Q,E)$, and let $E_R$ denote all relevant edges. Consider any
  graph $G^*=(P,Q,E^*)$. If for all $(A,B) \in \coll$, we have $W(A,B) \cap
  E^* \cap E_R \neq \phi$, then $G^*$ has a perfect matching.
\end{lemma}

\subsection{A Corollary of Bencz\'{u}r-Karger Sampling Scheme}

The Bencz\'{u}r-Karger sampling theorem~\cite{benczurkarger96} shows that for any
graph, a relatively small {\em non-uniform} edge sampling rate suffices to
ensure that every cut in the graph is hit by the sampled edges (i.e. it has a non-empty intersection)
with high probability. The sampling rate used for each edge $e$ inversely
depends on its strength, as defined below.

\label{sec:benczur-karger}
\begin{definition}~\cite{benczurkarger96} A $k$-strong component of a graph
  $H$ is a maximal vertex-induced subgraph of $H$ with edge-connectivity $k$.  The
  strength of an edge $e$ in a graph $H$ is the maximum value of $k$ such that
  a $k$-strong component contains $e$.
\end{definition}
\begin{definition}
  Given a graph $H=(V,E)$, let $H_{[j]} = (V,E_{[j]})$ denote the subgraph of $H$
  restricted to edges of strength $j$ or higher, where $j$ is some integer in
  $\{1, 2, \ldots, |V|\}$.
\end{definition}

It is easy to see that whenever a cut in a graph $H(V,E)$ contains an edge of strength $k$,
then the cut must contain at least $k$ edges. Furthermore, for any $1 < j \le |V|$, each connected
component of graph $H_{[j]}$ is contained inside some connected component of $H_{[j-1]}$.
The Bencz\'{u}r-Karger theorem utilizes these properties to show that it suffices to sample
each edge $e$ with probability $\Theta( \min\{ 1, \ln n / s_e \})$.

We now extend this sampling result to any collection of edge-sets for which there exists an injection
(one-one mapping) to cuts of comparable inverse strengths. The statement
of our theorem~\ref{thm:Benczur-Karger} closely mirrors the Bencz\'{u}r-Karger
sampling theorem, and the proof is also along the same general lines. However, the proof does not follow from the Bencz\'{u}r-Karger sampling theorem in a black-box fashion, so a proof is provided in appendix \ref{append:bk}
\begin{theorem}
\label{thm:Benczur-Karger}
Let $H(V,E)$ be any graph on $n$ vertices, and let ${\mathcal C}$ denote the set
of all possible edge cuts in $H$, and $\gamma \in (0,1]$ be a constant. Let
$H'$ be a subgraph of $H$ obtained by sampling each edge $e$ in $H$ with
probability
$ p_e = \min\left\{ 1, \frac{c \ln n}{\gamma s_e} \right\},$
where $s_e$ denotes the strength of edge $e$, and $c$ is a suitably large
constant. Further, let $\edgesets$ be a collection of subset of edges, and let $f$
be a one-one (not necessarily onto) mapping from $\edgesets$ to $\mathcal C$
satisfying $\sum_{e\in X} 1/s_e > \gamma\sum_{e\in f(X)} 1/s_e$ for all $X\in
\edgesets$. Then
$$ \sum_{ X \in \edgesets }  \Pr[ {\rm No~edge~in~}X~{\rm is~chosen~in~}H']
\leq
\frac{1}{n^2}.$$
\end{theorem}

The result below from~\cite{benczurkarger96} bounds the number of edges chosen
by the sampling in Theorem~\ref{thm:Benczur-Karger}.
\begin{theorem}
\label{thm:Benczur-Karger2}
Let $H(V,E)$ be any graph on $n$ vertices, and let $H'$ be a subgraph of $H$ obtained by
sampling each edge $e$ in $H$ with probability
$ p_e = \min\left\{ 1, \frac{c \ln n}{s_e} \right\},$
where $s_e$ denotes the strength of edge $e$, and $c$ is any constant. Then
with probability at least $1 - \frac{1}{n^2}$, the graph $H'$ contains at most
$c' n \ln n$ edges, where $c'$ is another suitably large constant.
\end{theorem}

We conclude with a simple property of integer multisets that we will use
later. A similar statement was used in \cite{karger-levine} (lemma 4.5). A proof is provided in appendix~\ref{append:strength} for completeness.
\begin{lemma}
\label{lem:strength}
Let $S_1$ and $S_2$ be two arbitrary multisets of positive integers such that $|S_1| > \gamma |S_2|$ for some $\gamma > 0$.
Then there exists an integer $j$ such that
$$ \sum_{i \ge j~ {\rm and}~ i \in S_1} {\frac{1}{i}} > \gamma \left( \sum_{i \ge j~ {\rm and}~ i \in S_2} {\frac{1}{i}} \right).$$
\end{lemma}

\section{Proportionate Uncrossing of Witness Sets}
\label{sec:proportion}
Consider a bipartite graph $G=(P,Q,E)$, with a non-negative weight function
$t$ defined on the edges. Assume further that we are given a set of ``relevant
edges'' $E_R \subseteq E$. We can extend the definition of $t$ to sets of
edges, so that $t(S) = \sum_{e\in S}t(e)$, where $S\subseteq E$.

\begin{definition}
\label{def:gamma_thick}
For any $\gamma >
0$ and $A\subseteq P, B\subseteq Q$, the pair $(A,B)$ is said to be
$\gamma$-thick with respect to $(G,t,E_R)$ if $t(W(A,B)\cap E_R) > \gamma
t(C(A,B))$, {\em i.e.}, the total weight of the relevant edges in $W(A,B)$ is {\em
  strictly} more than $\gamma$ times the total weight of $C(A,B)$. A set of
pairs $\coll = \{(A_1, B_1), (A_2, B_2), \ldots, (A_K,B_K)\}$ where each $A_i
\subseteq P$ and each $B_i \subseteq Q$ is said to be a $\gamma$-thick
collection with respect to $(G,t,E_R)$ if every pair $(A_i,B_i) \in \coll$ is
$\gamma$-thick.
\end{definition}

The quantities $G,t,$ and $E_R$ will be fixed for this section, and
for brevity, we will omit the phrase ``with respect to $(G,t,E_R)$'' in the
rest of this section.

Before defining proportionate uncrossings of witness sets, we will informally
point out the motivation for doing so. If a pair $(A,B)$ is $\gamma$-thick for
some constant $\gamma$, and if we know that a sampling process where edge $e$
is chosen with probability $t$ chooses some edge from $C(A,B)$ with high
probability, then increasing the sampling probability by a factor of
$1/\gamma$ should result in some relevant edge from $W(A,B)$ being chosen with
high probability as well, a fact that would be very useful in the rest of this
paper. The sampling sub-routines that we employ in the rest of this paper are
analyzed by using union-bound over all cuts, and in order to apply the same
union bound, it would be useful if each witness set were to correspond to a
unique cut. However, in figure~\ref{fig:cross}(a), we show two pairs $(A,B)$
and $(X,Y)$ which are both $(1/2)$-thick but correspond to the same cut; we
call this a ``crossing'' of the pairs $(A,B)$ and $(X,Y)$, drawing intuition
from the figure. In general, we can have many witness sets that map to the
same cut. We would like to ``uncross'' these witness sets by finding subsets
of each witness set that map to unique cuts, but there is no way to uncross
figure~\ref{fig:cross}(a) in this fashion. Fortunately, and somewhat
surprisingly, this is the worst case: any collection of $\gamma$-thick pairs
can be uncrossed into another collection such that all the pairs in the new
collection are also $\gamma$-thick (hence the term proportionate uncrossing),
every original witness set has a representative in this new collection, and no
more than two new pairs have the same cut. Figure~\ref{fig:cross}(b) shows two
$\frac{1}{2}$-thick pairs that can be uncrossed using a single $\frac{1}{2}$-thick
representative, $(A\cap X, B\cap Y)$. We will spend the rest of this section
formalizing the notion of proportionate uncrossings and proving their
existence. The uncrossing process is algorithmically inefficient, but we only
need to demonstrate existence for the purpose of this paper. The arguments in
this section represent the primary technical contribution of this paper; these
arguments apply to bipartite graphs in general (not necessarily regular), and
may be independently interesting. \vspace{-0.1in}
\begin{figure}[htbp]
  \centering
  \includegraphics[height=1.8in]{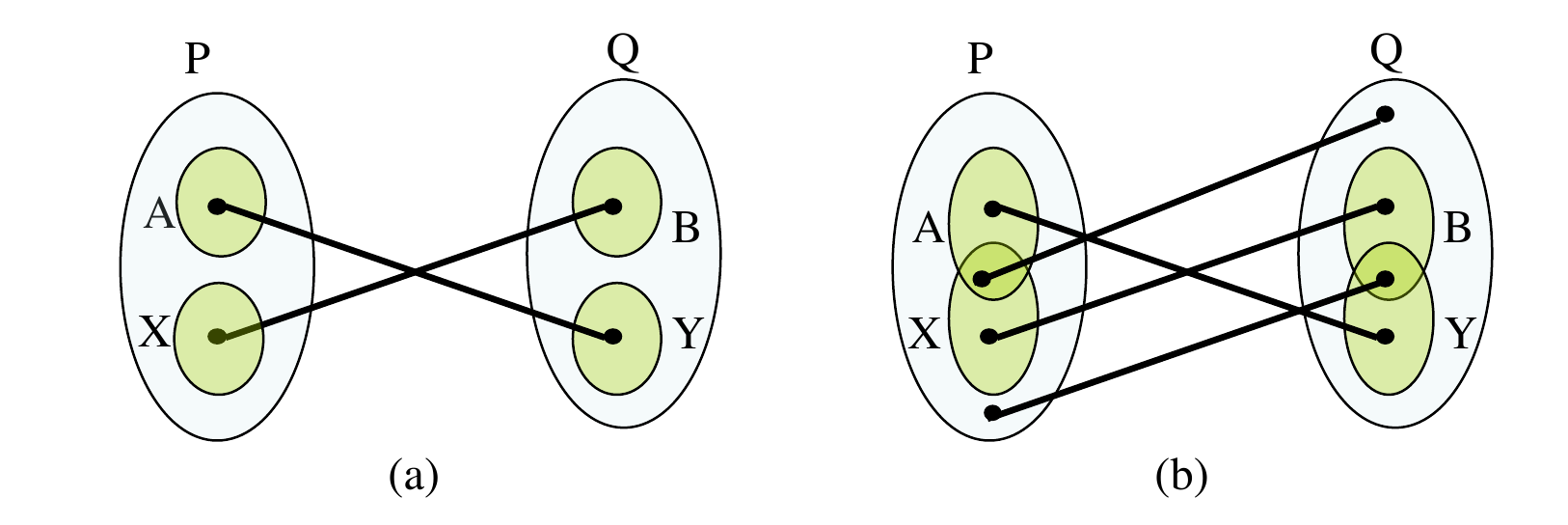}
  \vspace{-0.15in}
  \caption{ Both (a) and (b) depict two $\frac{1}{2}$-thick pairs $(A,B)$ and
    $(X,Y)$ that have different witness sets but the same cut (i.e. $W(A,B) \neq W(X,Y)$ but $C(A,B) = C(X,Y)$). The pairs in
    (a) can not be uncrossed, whereas the pairs in (b) can be uncrossed by
    choosing the single pair $(A\cap X, B\cap Y)$ as a representative.}
\label{fig:cross}
\end{figure}
\subsection{Proportionate Uncrossings: Definitions and Properties}
\label{sec:define-proportionate}
\begin{definition}
\label{define:uncross}
A $\gamma$-uncrossing of a $\gamma$-thick collection $\coll$ is another
$\gamma$-thick collection of pairs $\colltwo$ that satisfies the three
properties below:
\begin{description}
\item[P1:] For every pair $(A,B) \in \coll$ there exists a pair $(A',B') \in
  \colltwo$ such that $C(A',B') \subseteq C(A,B)$, and $W(A',B') \subseteq
  W(A,B)$. We will refer to $(A',B')$ as a representative of $(A,B)$.
\item[P2] For every $(A',B')\in \colltwo$, there exists $(A,B)\in\coll$ such
    that $C(A',B')\subseteq C(A,B)$.
  \item[P3:] {\em (Half-injectivity):} There can not be three distinct pairs
    $(A,B), (A',B'),$ and $(A'',B'')$ in $\colltwo$ such that $C(A, B) =
    C(A',B') = C(A'',B'')$.
\end{description}
\end{definition}

Since $\colltwo$ has the same (or larger) thickness as the thickness guarantee
that we had for $\coll$, it seems appropriate to refer to $\colltwo$ as a
proportionate uncrossing of $\coll$.

\begin{definition}
  \label{define:partial}
  A $\gamma$-partial-uncrossing of a $\gamma$-thick collection $\coll$ is
  another $\gamma$-thick collection of pairs $\colltwo$ which satisfies
  properties P1,P2 above but not necessarily P3.
\end{definition}
The following three lemmas follow immediately from the two definitions above,
and it will be useful to state them explicitly. Informally, the first says
that every collection is its own partial uncrossing, the second says that
uncrossings can be composed, and the third says that the union of the partial
uncrossings of two collections is a partial uncrossing of the union of the
collections.
\begin{lemma}
  \label{lem:self-uncross}
  If $\coll$ is a $\gamma$-thick collection, then $\coll$ is a
  $\gamma$-partial uncrossing of itself.
\end{lemma}
\begin{lemma}
\label{lem:compose}
If $\collpart$ is a $\gamma$-partial uncrossing of a $\gamma$-thick collection
$\coll$, and $\colltwo$ is a $\gamma$-uncrossing of $\collpart$, then
$\colltwo$ is also a $\gamma$-uncrossing of $\coll$.
\end{lemma}
\begin{lemma}
  \label{lem:union}
  If $\coll_1$ and $\coll_2$ are two $\gamma$-thick connections, $\colltwo_1$
  is a $\gamma$-partial-uncrossing of $\coll_1$, and $\colltwo_2$ is a
  $\gamma$-partial-uncrossing of $\coll_2$, then $\colltwo_1 \cup \colltwo_2$
  is a $\gamma$-partial-uncrossing of $\coll_1 \cup \coll_2$.
\end{lemma}

\subsection{Proportionate Uncrossings: An Existence Theorem}
\label{sec:exists-proportionate}
The main technical result of this section is the following:
\begin{theorem}
\label{thm:uncross}
  For every $\gamma$-thick collection $\coll$, there exists a
  $\gamma$-uncrossing of $\coll$. 
\end{theorem}
The proof is via induction over the ``largest cut'' corresponding to any pair
in the collection $\coll$; each inductive step ``uncrosses'' the witness sets
which corresponds to this largest cut. Before proving this theorem, we need to
provide several useful definitions and also establish a key lemma.

Define some total ordering $\prec$ over all subsets of $E$ which respects set
cardinality, so that if $|E_1| < |E_2|$ then $E_1 \prec E_2$. Overload
notation to use $C(\coll)$ to denote the set of cuts $\{C(A,B) : (A,B) \in
\coll$\}. Analogously, use $W(\coll)$ to denote the set of witness sets
corresponding to pairs in $\coll$. Since $C(A,B)$ may be equal to $C(A',B')$
for $(A,B) \neq (A',B')$, it is possible that $|C(\coll)|$ may be smaller than
$|\coll|$. In fact, if $\coll$ and $|C(\coll)|$ are equal, then $\coll$ is its
own $\gamma$-uncrossing and the theorem is trivially true. Similarly, it is
possible that $W(A,B)$ is equal to $W(A',B')$ for two different pairs $(A,B)$
and $(A',B')$ in $\coll$. However, suppose $W(A,B) = W(A',B')$ and $C(A,B) =
C(A',B')$ for two different pairs $(A,B)$ and $(A',B')$ in $\coll$. In this
case, we can remove one of the two pairs from the collection to obtain a new
collection $\coll'$; it is easy to see that a $\gamma$-uncrossing of $\coll'$
is also a $\gamma$-uncrossing of $\coll$. So we will assume without loss of
generality that for any two pairs $(A,B)$ and $(A',B')$ in $\coll$, either
$W(A,B) \neq W(A',B')$ or $C(A,B) \neq C(A',B')$; we will call this the {\em
  non-redundancy} assumption.

We will now prove a key lemma which contains the meat of the uncrossing
argument. When we use this lemma later in the proof of
theorem~\ref{thm:uncross}, we will only use the fact that there exists a
$\gamma$-partial-uncrossing of $\coll$, where $\coll$ satisfies the
preconditions of the lemma. However, the stronger claim of existence of a
$\gamma$-uncrossing does not require much additional work and appears to be an
interesting graph theoretic argument in its own right, so we prove this
stronger claim.

\begin{lemma}
  \label{lem:uncross}
  If $\coll$ is a $\gamma$-thick collection such that $|\coll| > 2$, $\coll$
  satisfies the non-redundancy assumption, and $C(\coll)$ contains a single
  set $S$, then there exists a $\gamma$-uncrossing $\colltwo$ of
  $\coll$. Further, for every pair $(A,B)\in \colltwo$, we have $C(A,B)
  \subset S$.
\end{lemma}
\begin{proofnoqed}
  Let $\coll = \{(A_1,B_1), (A_2,B_2), \ldots, (A_J,B_J)\}$. Since $C(A_i,B_i)
  = S$ for all $i$, we know by the non-redundancy assumption that $W(A_i,B_i)
  \neq W(A_{i'},B_{i'})$ for $i\neq i'$. We break the proof down into multiple
  stages.

  \begin{enumerate}
  \item {\em Definition of Venn witnesses and Venn cuts. } For any
    $J$-dimensional bit-vector $b\in\{0,1\}^J$, define
    $$A_{(b)} = \left(P
      \cap \left(\bigcap_{b_i = 1} A_i\right)\right) \setminus
    \left(\bigcup_{b_i=0} A_i\right) \mbox{, and similarly, } B_{(b)} =
    \left(Q\cap\left(\bigcap_{b_i = 1} B_i\right)\right) \setminus
    \left(\bigcup_{b_i=0} B_i\right).$$

    We overload notation and use $W_{(b)}$ to denote the witness set
    $W(A_{(b)},B_{(b)})$ and $C_{(b)}$ to denote the cut set
    $C(A_{(b)},B_{(b)})$. A node $u$ belongs to $A_{(b)}$ if it is in every
    set $A_i$ such that $b_i=1$ and not in any of the sets $A_i$ for which
    $b_i=0$. Thus, each $A_{(b)}$ corresponds to one of the regions in the
    Venn diagram of the sets $A_1, A_2, \ldots, A_J$, and the analogous
    statement holds for each $B_{(b)}$. Hence, we will refer to the sets
    $W_{(b)}$ and $C_{(b)}$ as the Venn-witness and the Venn-cut for $b$,
    respectively, and refer to the pair $(A_{(b)},B_{(b)})$ as a Venn
    pair. Also, we will use $\overline{b}$ to refer to a vector which differs
    from $b$ in every bit.

  \item {\em The special structure of Venn witnesses and Venn cuts.} Consider
    an edge $(u,v)$ that goes out of $A_{(b)}$. Suppose that edge goes to
    $B_{(d)}$ where $d\neq b$ and $d\neq \overline{b}$. Then there must exist
    $1\leq i,i' \leq J$ such that $b_i = d_i$ and $b_{i'} \neq d_{i'}$. Since
    $b_i = d_i$, either $u\in A_i, v\in B_i$ (if $b_i = d_i = 1$) or $u\not\in
    A_i, v\not\in B_i$ (if $b_i = d_i = 0$). In either case the edge $(u,v)$
    does not belong to the cut $C(A_i, B_i)$, and since all pairs in $\coll$
    have the same cut $S$, we conclude that $(u,v) \not\in S$. On the other
    hand, since $b_{i'} \neq d_{i'}$, either $u\in A_{i'}, v\not\in B_{i'}$
    (if $b_{i'} = 1, d_{i'} = 0$) or $u\not\in A_{i'}, v\in B_{i'}$ (if
    $b_{i'} = 0, d_{i'} = 1$). In either case the edge $(u,v)$ belongs to the
    cut $C(A_{i'},B_{i'})$ and hence to $S$, which is a contradiction. Thus,
    {\em any edge from $A_{(b)}$ goes to either $B_{(b)}$ or
      $B_{(\overline{b})}$}.

    If the edge $(u,v)$ goes to $B_{(b)}$ then it does not belong to any
    witness set in $W(\coll)$, any Venn witness set, any Venn cut, or $S$. If
    $(u,v)$ goes to $B_{(\overline{b})}$ then it belongs to $S$, to the Venn
    witness set $W_{(b)}$, to the Venn cuts $C_{(b)}$ and $C_{(\overline b)}$,
    and to no other Venn witness set or Venn cut. This edge also belongs to
    $W(A_i,B_i)$ for all $i$ such that $b_i = 1$. These observations, and the
    definitions of Venn witnesses, cuts, and pairs easily lead to the
    following consequences:
    \begin{equation}
      \label{eq:venn-witness-disjoint}
      W_{(b)} \cap W_{(d)} =\emptyset \mbox{ if } b \neq d,
    \end{equation}
    \begin{equation}
      \label{eq:venn-witness-union}
      W(A_i,B_i) = \bigcup_{b\in\{0,1\}^J : b_i = 1}W_{(b)},
    \end{equation}
    \begin{equation}
      \label{eq:venn-cut-symmetry}
      C_{(b)} = C_{(\overline{b})},
    \end{equation}
    \begin{equation}
      \label{eq:venn-cut-disjoint}
      C_{(b)} \cap C_{(d)} = \emptyset \mbox{ if } b\neq d \mbox { and } b\neq
    \overline d,
    \end{equation}
    \begin{equation}
      \label{eq:venn-cut-union}
      (\forall i, 1\le i \le J): S = \bigcup_{b\in\{0,1\}^J : b_i = 1}C_{(b)},
    \end{equation}
 and finally,
 \begin{equation}
   \label{eq:venn-witness-cut}
   W_{(b)} \cup W_{(\overline b)} = C_{(b)}.
 \end{equation}

\item{\em The collection $\colltwo$.} Define $\colltwo$ to consist of all
  $\gamma$-thick Venn pairs $(A_{(b)}, B_{(b)})$ where $b$ is not the all zero
  vector.

\item{\em Proving that $\colltwo$ is a $\gamma$-uncrossing of $\coll$.}  {\bf
    (P1):} Fix some $i, 1\le i\le J$. Since $\coll$ is a $\gamma$-thick
  collection, it follows from the definition that $(A_i, B_i)$ must be a
  $\gamma$-thick pair. From equations~\ref{eq:venn-witness-union}
  and~\ref{eq:venn-witness-disjoint}, we know that $t(W(A_i,B_i) \cap E_R) =
  \sum_{b\in\{0,1\}^J : b_i = 1} t(W_{(b)} \cap E_R)$. We also know, from
  equations~\ref{eq:venn-cut-disjoint} and~\ref{eq:venn-cut-union}, that $t(S)
  = \sum_{b\in\{0,1\}^J : b_i = 1}t(C_{(b)})$. Hence, there must be some $b\in
  \{0,1\}^J$ such that $b_i = 1$ and $(A_{(b)},B_{(b)})$ is $\gamma$-thick,
  which in turn implies that $(A_{(b)},B_{(b)})$ is in $\colltwo$. This is the
  representative of $(A_i,B_i)$ and hence $\colltwo$ satisfies P1.  {\bf
    (P2):} This follows trivially from equation~\ref{eq:venn-cut-union}.  {\bf
    (P3):} From equation~\ref{eq:venn-cut-disjoint} we know that there are
  only two possible Venn pairs (specifically, $(A_{(b)},B_{(b)})$ and
  $(A_{(\overline b)},B_{(\overline b)})$) that have the same non-empty cut
  $C_{(b)}$. Observe that our definition of $\gamma$-thickness involves
  ``strict inequality'', and hence Venn pairs where the Venn witness set and
  the Venn cut are both empty can't be $\gamma$-thick and can't be in
  $\colltwo$.

\item {\em Proving that $C(X,Y) \subset S$ for all pairs $(X,Y)\in \colltwo$.}
  Any cut $C(A,B) \in C(\colltwo)$ is of the form $C_{(b)}$ for some
  $J$-dimensional bit vector $b$. Each $C_{(b)} \subseteq S$, from
  equation~\ref{eq:venn-cut-union}. We will now show that this containment is
  strict. Suppose not, {\em i.e.}, there exists some $C_{(b)} = S$. By
  equation~\ref{eq:venn-cut-symmetry}, $C_{(\overline b)} = S$ as well. Since
  $J > 2$, either $b$ or $\overline b$ must have two bits that are set to 1;
  without loss of generality, assume that $b_1 = b_2 = 1$. From
  equations~\ref{eq:venn-witness-disjoint} and~\ref{eq:venn-witness-cut}, we
  know that $C_{(b)}$ (and hence $S$) is the disjoint union of $W_{(b)}$ and
  $W_{(\overline b)}$. Any edge in $W_{(b)}$ must belong to both $W(A_1,B_1)$
  and $W(A_2, B_2)$, whereas any edge in $W_{(\overline b)}$ can not belong to
  either $W(A_1,B_1)$ or $W(A_2, B_2)$. Hence, $W(A_1, B_1) = W(A_2, B_2) =
  W_{(b)}$ which contradicts the non-redundancy assumption on
  $\coll$. Therefore, we must have $C_{(b)} \subset S$. $\qed$
\end{enumerate}
\end{proofnoqed}

\begin{proofof}{Theorem~\ref{thm:uncross}}
  The proof will be by induction over the largest set in $C(\coll)$ according
  to the ordering $\prec$. Let $M(\coll)$ denote this largest set.

  For the base case, suppose $M(\coll)$ is the smallest set $S$ under the
  ordering $\prec$. Then $S$ must be singleton, $C(\coll)$ must have just a
  single set $S$, and $W(\coll)$ must also have a single witness set, which
  must be the same as $S$ since $\coll$ is $\gamma$-thick. By the
  non-redundancy assumption, $\coll$ must have at most one pair, and is its
  own $\gamma$-uncrossing.

  For the inductive step, consider any possible cut $S$ and assume that the
  theorem is true when $M(\coll) \prec S$. We will show that the theorem is
  also true when $M(\coll) = S$, which will complete the inductive proof.

  Suppose there is a unique $(A,B) \in \coll$ such that $C(A,B) =
  S$. Intuitively, one would expect this to be the easy case, since there is
  no ``uncrossing'' to be done for $S$, and indeed, this case is quite
  straightforward. Define $\coll' = \coll - (A,B)$. Let $\colltwo'$ denote a
  $\gamma$-uncrossing of $\coll$, which is guaranteed to exist by the
  inductive hypothesis. Since $\colltwo'$ is $\gamma$-thick, so is $\colltwo =
  \colltwo' \cup \{(A,B)\}$. The pair $(A,B)$ clearly has a representative in
  $\colltwo$ (itself), and any $(A',B') \in \coll - (A,B)$ has a
  representative in $\colltwo'$ and hence also in $\colltwo$. Thus, $\colltwo$
  satisfies property P1 for being a $\gamma$-uncrossing of $\coll$.  Every set
  in $C(\colltwo')$ is a subset of some cut in $C(\coll')$ (by property P2)
  and $C(A,B)$ is also in $C(\coll)$, and hence $\colltwo$ satisfies property
  P2 for being a $\gamma$-uncrossing of $\coll$. Every set in $\colltwo'$ is
  smaller than $C(A,B)$ according to $\prec$ and $\colltwo'$ satisfies
  property P3. Hence, $\colltwo$ also satisfies property P3. Thus, $\colltwo$
  is a $\gamma$-uncrossing of $\coll$. If there are exactly two distinct pairs
  $(A,B)$ and $(A',B')$ in $\coll$ such that $C(A,B) = C(A',B') = S$, then the
  same argument works again, except that $\coll' = \coll \setminus\{ (A,B)
  ,(A',B')\}$ and $\colltwo = \colltwo' \cup \{(A,B),(A',B')\}$.

  We now need to tackle the most interesting case of the inductive step, where
  there are more than two pairs in $\coll$ that correspond to the same cut
  $S$. Write $\coll = \coll_1 \cup \coll_2$ where $C(A,B) \prec S$ for all
  $(A,B) \in \coll_1$ and $C(A,B) = S$ for all $(A,B) \in \coll_2$. Recall
  that for two different pairs $(A,B)$ and $(A',B')$ in $\coll_2$, we must
  have $W(A,B) \neq W(A',B')$ by the non-redundancy assumption. From
  lemma~\ref{lem:uncross}, there exists a $\gamma$-partial-uncrossing, say
  $\collpart_2$, of $\coll_2$ with the property that for every set $S' \in
  C(\collpart_2)$, we have $S' \subset S$, and hence $S' \prec S$. By
  lemma~\ref{lem:self-uncross}, we know that $\coll_1$ is its own
  $\gamma$-partial-uncrossing. Further, by definition of $\coll_1$, every set
  $S' \in C(\coll_1)$ must satisfy $S' \prec S$. Define $\collpart = \coll_1
  \cup \collpart_2$. By lemma~\ref{lem:union}, $\collpart$ is a
  $\gamma$-partial-uncrossing of $\coll_1 \cup \coll_2$, {\em i.e.}, of
  $\coll$. Further, for every cut $S' \in C(\collpart)$, we have $S' \prec
  S$. Hence, by our inductive hypothesis, there exists a $\gamma$-uncrossing
  of $\collpart$; let $\colltwo$ be a $\gamma$-uncrossing of $\collpart$. By
  lemma~\ref{lem:compose}, $\colltwo$ is also a $\gamma$-uncrossing of
  $\coll$, which completes the inductive proof.
\end{proofof}

\begin{remark} \label{rm:other-approach}
An alternate approach to relating cuts and witness sets is to suitably
modify the proof of the
Bencz\'{u}r-Karger sampling theorem, circumventing the need for the
proportionate uncrossing
theorem. The idea is based on the observation that Karger's sampling
theorem also holds for
vertex cuts in graphs. Since Bencz\'{u}r-Karger sampling theorem is
proved using multiple
invocations of Karger's sampling theorem, it is possible to set up a
correspondence between
cuts and witness sets using a vertex-cut version of the
Bencz\'{u}r-Karger sampling theorem.
However, we prefer to use here the approach based on the proportionate
uncrossing theorem
as it is an interesting combinatorial statement in its own right.
\end{remark}

\section{An $\tilde{O}(n^{1.5})$ Time Algorithm for Finding a Perfect Matching}
\label{sec:algo}

We present here an $\tilde{O}(n^{1.5})$ time randomized algorithm to find a
perfect matching in a given $d$-regular bipartite graph $G(P,Q,E)$ on $2n$
vertices.  Throughout this section, we follow the convention that for any pair
$(A,B)$, the sets $C(A,B)$ and $W(A,B)$ are defined with respect to the graph
$G$.  Our starting point is the following theorem, established by Goel,
Kapralov, and Khanna~\cite{gkk:rbp08}\footnote{Part 1 of
  theorem~\ref{thm:GKK09} corresponds to theorem 2.3 in~\cite{gkk:rbp08}, part
  2 is proved as part of the proof of theorem 2.1 in~\cite{gkk:rbp08}, and
  part 3 combines remark 2.5 in~\cite{gkk:rbp08} with Karger's sampling
  theorem~\cite{kar94a}.}

\begin{theorem}
\label{thm:GKK09}
Let $G(P,Q,E)$ be a $d$-regular bipartite graph, $\epsilon$ any number in
$(0,\frac{1}{2})$, and $c$ a suitably large constant that depends on
$\epsilon$. There {\em exists} a decomposition of $G$ into $k = O(n/d)$
vertex-disjoint bipartite graphs, say $G_1=(P_1,Q_1,E_1), G_2=(P_2,Q_2,E_2),
\ldots, G_k=(P_k,Q_k,E_k)$, such that

\begin{enumerate}
\item
Each $G_i$ contains at least $d/2$ perfect matchings, and the minimum cut in each $G_i$ is $\Omega(d^2/n)$.

\item Let $\coll$ denote the set of relevant pairs with respect to this
  decomposition, and $E_R$ denote the set of relevant edges.  Then for each
  $(A,B)$ in $\coll$, we have $|W(A,B) \cap E_R| \geq \frac{1}{2} |C(A,B)|$.

\item Let $G'(P,Q,E')$ be a random graph generated by sampling the edges of
  $G$ {\em uniformly} at random with probability $p = \frac{c n \ln n}{d^2}$.
  Then with probability at least $1 - 1/n$, for every pair $(A,B) \in
  \coll$,
$$|W(A,B) \cap E' \cap E_R| > (1 - \epsilon)p |W(A,B) \cap E_R| > \left(
  \frac{1- \epsilon}{2(1+\epsilon)} \right) |C(A,B)\cap E'|.$$
\end{enumerate}
\end{theorem}

The last condition above says that in addition to all cuts, all relevant
witness edge sets are also preserved to within $(1 \pm \epsilon)$ of their
expected value in $G'$, with high probability. We emphasize here that the
decomposition highlighted in Theorem~\ref{thm:GKK09} will be used only in the
analysis of our algorithm; the algorithm itself is oblivious to this
decomposition.

Our algorithm consists of the following three steps.

\begin{description}

\item[(S1)] Generate a random graph $G'=(P,Q,E')$ by sampling edges of $G$
  {\em uniformly} at random with probability $p = \frac{c_1 n \ln n}{d^2}$
  where $c_1$ is a constant as in Theorem~\ref{thm:GKK09}\footnote{The time
    required for this sampling is proportional to the number of edges chosen,
    assuming the graph is presented in an adjacency list representation with
    each list stored in an array.}.  We choose $\epsilon$ to be any fixed
  constant not larger than $0.2$.

\item[(S2)] The graph $G'$ contains $O( \frac{n^2 \ln n}{d})$ edges w.h.p. We
  now run the Bencz\'{u}r-Karger sampling algorithm~\cite{benczurkarger96} that
  takes $O(|E'|\ln^2 n)$ time to compute the strength $s_e$ of every edge $e$,
  and samples each edge $e$ with probability $p_e$;\footnote{In fact, this
    sampling algorithm computes an upper bound on $s_e$, but this only affects
    the running time and the number of edges sampled by a constant factor.}
  here $p_e$ is as given by Theorem~\ref{thm:Benczur-Karger} with $\gamma =
  1/3$. We show below that w.h.p. the graph $G^{''} = (P,Q,E^{''})$ obtained
  from this sampling contains a perfect matching.

\item[(S3)] Finally, we run the Hopcroft-Karp algorithm to obtain a maximum
  cardinality matching in $G^{''}$ in $O(n^{1.5} \ln n)$ time since by
  Theorem~\ref{thm:Benczur-Karger2}, $G^{''}$ contains $O(n \ln n)$ edges w.h.p.
\end{description}

\paragraph{Running time:} With high probability, the running time of this
algorithm is bounded by $O(\frac{n^2}{d}\ln^3 n + n^{1.5}\ln n)$. Since we can
always use the algorithm of Cole, Ost, and Schirra~\cite{cos:regular2001}
instead, the final running time is $O(\min\{m, \frac{n^2}{d}\ln^3 n +
n^{1.5}\ln n\})$. This reduces to $O(m)$ if $d \le \sqrt{n} \ln n $; to
$O(n^{1.5}\ln n)$ when $d \ge \sqrt{n}\ln^2 n$; and to at most $O((n\ln
n)^{1.5})$ in the narrow range $ \sqrt{n} \ln n < d < \sqrt{n}\ln^2 n$.

\paragraph{Correctness:} To prove correctness, we need to show that $G^{''}$
contains a perfect matching w.h.p.
\begin{theorem}\label{thm:correctness}
  The graph $G^{''}$ contains a perfect matching with probability $1 -
  O(1/n)$.
\end{theorem}
\begin{proof}
  Consider the decomposition defined in Theorem~\ref{thm:GKK09}.  Let $\coll$
  denote the set of relevant pairs with respect to this decomposition, and let
  $E_R$ denote the set of all relevant edges with respect to this
  decomposition.  We will now focus on proving that, with high probability,
  for every $(A,B) \in \coll$, $W(A,B) \cap E_R \cap E^{''}\neq \emptyset$; by
  Lemma~\ref{lem:relevant-only}, this is sufficient to prove the theorem.

  For convenience, define $W'(A,B) = W(A,B)\cap E'$ and $C'(A,B) = C(A,B) \cap
  E'$. Assume for now that the low-probability event in
  Theorem~\ref{thm:GKK09} does not occur. Thus, by choosing $\epsilon \le
  0.2$, we know that for $\gamma = 1/3$, every relevant pair $(A,B) \in \coll$
  satisfies
  $ |W'(A,B) \cap E_R| > \gamma |C'(A,B)|.$
  
  Let $s'_e$ denote the strength of $e$ in $G'$. Recall that
  $G'_{[j]}=(V,E'_{[j]})$ is the graph with the same vertex set as $G'$ but
  consisting of only those edges in $E'$ which have strength at least
  $j$. Define $W'_{[j]}(A,B)$ to be the set of all edges in $W'(A,B) \cap E'_{[j]}$;
  define $C'_{[j]}(A,B)$ analogously. Define $t(e) =
  1/s'_e$. Since $ |W'(A,B) \cap E_R| > \gamma |C'(A,B)|$, by
  Lemma~\ref{lem:strength}, there must exist a $j$ such that
  $$\sum_{e\in
    (W'(A,B) \cap E_R), s'_e \ge j} \frac{1}{s'_e} > \gamma \sum_{e\in C'(A,B), s'_e
    \ge j} \frac{1}{s'_e} > 0,$$
  which implies that $(A,B)$ is $\gamma$-thick with
  respect to $(G'_{[j]}, t, E_R)$, as defined in Definition~\ref{def:gamma_thick}.
  Partition $\coll$ into $\coll_{[1]}$,
  $\coll_{[2]}, \ldots, \coll_{[n]}$, such that if $(A,B)\in \coll_{[j]}$ then
  $(A,B)$ is $\gamma$-thick with respect to $(G'_{[j]}, t, E_R)$, breaking
  ties arbitrarily if $(A,B)$ can belong to multiple $\coll_{[j]}$. Consider
  an arbitrary non-empty $\coll_{[j]}$. Let $\colltwo$ represent a
  $\gamma$-uncrossing of $\coll_{[j]}$, as guaranteed by
  Theorem~\ref{thm:uncross}. By property P3, no three pairs in a
  $\gamma$-uncrossing can have the same cut; partition $\colltwo$ into
  $\colltwo_1$ and $\colltwo_2$ such that every pair $(A,B) \in \colltwo_1$
  has a unique cut $C'_{[j]}(A,B)$ and the same holds for $\colltwo_2$. We
  focus on $\colltwo_1$ for now. For any $(A,B)\in\colltwo_1$, define $Y(A,B)
  = W'_{[j]}(A,B)\cap E_R$. Define $\edgesets = \{Y(A,B): (A,B)\in
  \colltwo_1\}$. For any $X\in \edgesets$, define $f(X) = C'_{[j]}(A,B)$ for
  some arbitrary $(A,B)\in \colltwo_1$ such that $X=Y(A,B)$. The function $f$
  is one-one by construction, and since $(A,B)$ is $\gamma$-thick, we know
  that $\sum_{e\in X} 1/s'_e > \gamma \sum_{e\in f(X)} 1/s'_e$. Thus,
  $\edgesets$ satisfies the preconditions of
  Theorem~\ref{thm:Benczur-Karger}. Further, the sampling probability $p_e$ in
  step {\bf(S2)} of the algorithm is chosen to correspond to
  $\gamma=1/3$. Thus, with probability at least $1-1/n^2$, $X\cap E^{''}$ is
  non-empty for all $X\in \edgesets$, {\em i.e.}, $W'_{[j]}(A,B) \cap E_R \cap
  E^{''}\neq \emptyset$ for all $(A,B) \in \colltwo_1$. Since $G'_{[j]}$ is a
  subgraph of $G'$, we can conclude that $W'(A,B) \cap E_R \cap E^{''}\neq
  \emptyset$ for all $(A,B) \in \colltwo_1$ with probability at least
  $1-1/n^2$.

  Since the analogous argument holds for $\colltwo_2$, we obtain $W'(A,B) \cap
  E_R \cap E^{''}\neq \emptyset$ for all $(A,B) \in \colltwo$ with probability
  at least $1-2/n^2$. Since $\colltwo$ is a $\gamma$-uncrossing of
  $\coll_{[j]}$, we use property P1 to conclude that $W'(A,B) \cap E_R \cap
  E^{''}\neq \emptyset$ for all $(A,B) \in \coll_{[j]}$, again with
  probability at least $1-2/n^2$. Applying the union bound over all $j$, we
  further conclude that $W'(A,B) \cap E_R \cap E^{''}\neq \emptyset$ for all
  $(A,B) \in \coll$ with probability at least $1-2/n$. As mentioned before,
  this suffices to prove that $G^{''}$ has a perfect matching with probability
  at least $1-2/n$, by Lemma~\ref{lem:relevant-only}.  We assumed that
  condition 3 in theorem~\ref{thm:GKK09} is satisfied; this is violated with
  probability at most $\frac{1}{n}$, which proves that $G^{''}$ has a perfect
  matching with probability at least $1-{3\over n}$.
\end{proof}

As presented above, the algorithm takes time $\min\{\tilde{O}(n^{1.5}),
O(m)\}$ with high probability, and outputs a perfect matching with probability
$1- O(1/n)$. We conclude with two simple observations. First.  it is easy to
convert this into a Monte Carlo algorithm with a worst case running-time of
$\min\{\tilde{O}(n^{1.5}), O(m)\}$, or a Las Vegas algorithm with an expected
running-time of $\min\{\tilde{O}(n^{1.5}), O(m)\}$.  If either the sampling
process in steps {\bf (S1)} or {\bf (S2)} returns too many edges, or step {\bf
  (S3)} does not produce a perfect matching, then (a) abort the computation to
get a Monte Carlo algorithm, or (b) run the $O(m)$ time algorithm of Cole,
Ost, and Schirra~\cite{cos:regular2001} to get a Las Vegas algorithm. Second,
by choosing larger constants during steps {\bf (S1)} and {\bf (S2)}, it is
easy to amplify the success probability to be at least $1-O({1\over n^j})$ for
any fixed $j \ge 1$.

\section{An Improved $O\left(\min\{nd, (n^2\ln^3 n)/d\}\right)$ Bound on the
  Runtime} \label{sec:improved-runtime} 
 In this section we give an improved analysis of the runtime of the Hopcroft-Karp algorithm on the subsampled graph, ultimately leading to a bound of 
 $O\left(\min\{nd, (n^2\ln^3 n)/d\}\right)$ for our algorithm. The main ingredients of our analysis are (1) a decomposition of the graph $G$ into $O(n/d)$ vertex-disjoint $\Omega(d)$-edge-connected subgraphs,  (2) a modification of the uncrossing argument that reveals properties of sufficiently unbalanced witness sets in the sampled graph obtained in step \textbf{S2}, and (3) an upper bound on 	length of the shortest augmentating path in the sampled graph relative to any matching of size smaller than $n-2n/d$. 

\subsection{Combinatorial uncrossings}  
Theorem \ref{thm:ss} below, which we state for general bipartite graphs, requires a variant of the uncrossing theorem that we formulate now. We introduce the definition of combinatorial uncrossings:
\begin{definition}
\label{define:comp-uncross}
Let $\coll$ be any collection of pairs $(A, B), A\subseteq P, B\subseteq Q$. A combinatorial uncrossing of $\coll$ is a tuple $(\colltwo, \mathcal{I})$, where $\colltwo$ is another collection and $\mathcal{I}$ is a mapping from $\coll$ to subsets of $\colltwo$, such that the following properties are satisfied: 
\begin{description}
\item[Q1:] For all $(A, B)\in \coll$
\begin{enumerate}
\item $\left\lbrace W(A', B')\right\rbrace_{(A', B')\in \mathcal{I}(A, B)}$ are disjoint;
\item $\left\lbrace C(A', B')\right\rbrace_{(A', B')\in \mathcal{I}(A, B)}$ are disjoint;
\item $\left\lbrace A'\cup B'\right\rbrace_{(A', B')\in \mathcal{I}(A, B)}$ are disjoint;
\item $A'\subseteq A, B'\subseteq B$ for all $(A', B')\in \mathcal{I}(A, B)$;
\item 
\begin{equation*}
\begin{split}
W(A, B)=\bigcup_{(A', B')\in \mathcal{I}(A, B)} W(A', B')\\
C(A, B)=\bigcup_{(A', B')\in \mathcal{I}(A, B)} C(A', B').
\end{split}
\end{equation*}
\end{enumerate}
\item[Q2:](Half-injectivity) There cannot be three distinct pairs $(A, B), (A', B'), (A'', B'')$ in $\colltwo$ such that $C(A, B)=C(A', B')=C(A'', B'')$.
\end{description}
\end{definition}
The proof of existence of combinatorial uncrossings is along the lines of the proof of existence of $\gamma$-thick uncrossings, so we omit it here.

For a graph $H$ we denote $W_H(A, B)=W(A,B)\cap E(H)$ and $C_H(A, B)=C(A, B)\cap E(H)$, and omit the subscript when the
underlying graph is fixed.
\begin{theorem} \label{thm:ss}
Let $G^*$ be a graph obtained by sampling edges uniformly at random with
probability $p$ from a bipartite graph $G=(P, Q, E)$ on $2n$ vertices with a minimum cut of size
$\kappa$. Then there
exists a constant $c>0$ such that for all $\e>0$ if $p > \frac{c \ln n}{\e^2 \kappa}$
then w.h.p. for all $A \subseteq P$, and $B \subseteq Q$, we have
\begin{equation*}
\begin{split}
p|W_G(A,B)| - \e p| C_G(A,B) |\leq |W_{G^*}(A,B) |\leq p|W_G(A,B)| + \e p| C_G(A,B) |.
\end{split}
\end{equation*}
\end{theorem}
\begin{proof}
Define $\coll$ as the set of pairs $(A, B), A\subseteq P\cap V(G), B\subseteq Q\cap V(G)$. Denote a combinatorial uncrossing of $\coll$ by $(\colltwo, \mathcal{I})$.
We first prove the statement for pairs from $\colltwo$, and then extend it to pairs from $\coll$ to obtain the desired result.

Consider a pair $(A, B)\in \colltwo$.  Denote $\Delta_G(A, B)=|W_{G^*}(A, B)|-p|W_G(A, B)|$. We shall write $W(A, B)$ and $C(A, B)$ instead of $W_G(A, B)$ and $C_G(A, B)$ in what follows for brevity. We have by Chernoff bounds that for a given pair $(A, B)\in \colltwo$ 
\begin{equation*}
\begin{split}
\prob\left[\left|\Delta_{G}(A, B)\right|>\e p |C(A, B)|\right]<\exp\left[-\left(\frac{\e|C(A, B)|}{|W(A, B)|}\right)^2 \frac{p|W(A, B)|}{2} \right]\\
\leq \exp\left[-\e^2 \left( \frac{p|C(A, B)|}{2} \right) \right]
\end{split}
\end{equation*}
since $|C(A, B)|\geq |W(A, B)|$. Since $\colltwo$ satisfies Q2, we get that
\begin{equation*}
\begin{split}
&\prob\left[\exists (A, B)\in \colltwo: \left|\Delta_G(A, B)\right|>\e p|C(A, B)|\right]\\
&<\sum_{W(A, B)\in W(\colltwo)} \exp\left[-\e^2 p|C(A, B)|/2\right]\leq 2\sum_{C(A, B)\in C(\colltwo)} \exp\left[-\e^2 p|C(A, B)|/2\right]=O(n^{-r})
\end{split}
\end{equation*}
for $c=2(r+2)$ by Corollary 2.4 in \cite{kar94a}. This implies that for $c \ge 2(r+2)$ we have with probability $1-O(n^{-r})$ for all $(A, B)\in \colltwo$
\begin{equation}\label{ab-bound}
|\Delta_G(A, B)| \leq \e p| C(A,B) |.
\end{equation} 

Now consider any pair $(A, B)\in \coll$. Summing \eqref{ab-bound} over all $(A', B')\in \mathcal{I}(A,B)$ and using properties Q1.1-5, we get
\begin{equation*} 
\begin{split} 
|\Delta_G(A, B)|&\leq \sum_{(A', B')\in \mathcal{I}(A, B)}\e p|C(A', B')|=\e p|C(A, B)|,
\end{split}
\end{equation*} 
for all $(A, B)\in \coll$ as required.
\end{proof}

\subsection{Decomposition of the graph $G$} \label{subsec:decomp}
Corollary \ref{cor:unbalanced-ws}, which relates
the size of sufficiently unbalanced witness sets in the sampled graph to the
size of the corresponding cuts is the main result of this subsection. It follows from theorem~\ref{thm:ss} and a stronger
(than~\cite{gkk:rbp08}) decomposition of bipartite $d$-regular graphs that we outline now.
\begin{theorem} \label{thm:decomposition}
Any $d$-regular graph $G$ with $2n$ vertices can be decomposed into
 vertex-disjoint induced subgraphs $G_1=(P_1,Q_1,E_1), G_2 = (P_2,Q_2,E_2), ...., G_k =(P_k,Q_k,E_k)$, where $k\leq 4n/d+1$, that satisfy the following properties:

\begin{enumerate}
 \item The minimum cut in each $G_i$ is at least $d/8$.


 \item $\sum_{i=1}^{k+1}  | \delta_G (V(G_i)) | \leq 2n$.
\end{enumerate}
\end{theorem}

To prove Theorem \ref{thm:decomposition}, we give a procedure that decomposes the graph $G$ into vertex-disjoint induced subgraphs $G_1(P_1, Q_1, E_1)$, $G_2(P_2, Q_2, E_2),\ldots, G_k(P_k, Q_k, E_k)$, $k\leq 4n/d+1$ such that the min-cut in $G_j$ is at least $d/8$ and at most $n$ edges run between pieces of the decomposition. 

The procedure is as follows. Initialize $H_1:=G$, and set $i:=1$.
\begin{enumerate}
\item Find a smallest proper subset $X_i\subset V(H_i)$ such that $|\delta_{H_i}(X_i)|< d/4$. If no such set exists, define $G_i$ to be the graph $H_i$ and terminate. 
\item Define $G_i$ to be the subgraph of $H_i$ induced by vertices in $X_i$, i.e. $X_i=P_i\cup Q_i=V(G_i)$. Also, define $H_{i+1}$ to be the graph $H_i$ with vertices from $X_i$ removed.
\item Increment $i$ and go to step 1.	
\end{enumerate}

We now prove that the output of the decomposition procedure satisfies the properties claimed above.
\begin{lemma}\label{lm:mincut}
The min-cut in $G_i$ is greater than $d/8$.
\end{lemma}
\begin{proof}
If $G_i$ contains a single vertex the min-cut is infinite by definition, so we assume wlog that $G_i$ contains at least two vertices.
The proof is essentially the same as the proof of property \textbf{P1} of the decomposition procedure in \cite{gkk:rbp08} (see Theorem 2.4).

Suppose that there exists a cut $(V, V^c)$ in $G_i$ where $V\subset V(G_i)$ and
$V^c=V(G_i)\setminus V$, such that $|\delta_{G_i}(V)|\leq d/8$ (note that it is possible that $V\cap P_i\neq \emptyset$ and $V\cap Q_i\neq \emptyset$).  We have $|\delta_{H_i}(V)\setminus \delta_{G_i}(V)|+|\delta_{H_i}(V^c)\setminus \delta_{G_i}(V^c)|< d/4$ by the choice of $X_i$ in (1). Suppose without loss of generality that $|\delta_{H_i}(V)\setminus \delta_{G_i}(V)|<d/8$.
Then $|\delta_{H_i}(V)|< d/4$ and $V\subset X_i$, which contradicts the choice of $X_i$ as the smallest cut of value at most $d/4$ in step (1) of the procedure.
\end{proof}

\begin{lemma} \label{lm:edges-removed}
The number of steps in the decomposition procedure is $k\leq 4n/d$, and at most $n$ edges are removed in the process.
\end{lemma}
\begin{proof}
We call a vertex $v\in V(G_i)$ {\em bad} if its degree in $G_i$ is smaller than $d/2$. Note that for each $1\leq i\leq k$ either $G_i$ contains a bad vertex or $V(G_i)\geq d$. 

Note that since strictly fewer than $d/4$ edges are removed in each iteration, the number of bad vertices created in the first $j$ iterations is strictly less than $j(d/4)/(d/2)=j/2$.  Hence, during at least half of the $j$ iterations at least $d$ vertices were removed from the graph, i.e.
\begin{equation*}
\sum_{i=1}^j |V(G_i)|\geq (j/2)\cdot d=jd/2.
\end{equation*}
This implies that the process terminates in at most $4n/d$ steps, and the number of edges removed is at most $(4n/d)\cdot d/4=n$.
\end{proof}
\begin{proofof}{Theorem \ref{thm:decomposition}}
The proof follows by putting together lemmas \ref{lm:mincut} and \ref{lm:edges-removed}.
\end{proofof}

We overload notation here by denoting $W(B, A)=W(P\setminus A, Q\setminus B)=C(A, B)\setminus W(A, B)$ for $A\subseteq P, B\subseteq Q$. The main result of this subsection is
\begin{corollary} \label{cor:unbalanced-ws}
Let $G^*=(P, Q, E^*)$ be a graph obtained by sampling the edges of a $d$-regular bipartite graph $G=(P, Q, E)$ on $2n$ vertices independently with probability $p$. There exists a constant $c>0$ such that if $p>\frac{c\ln n}{\e^2 d}$ then whp for all pairs $(A, B), A\subseteq P, B\subseteq Q$, $|A|\geq |B|+2n/d$ one has that 
$|W(A,B)\cap E^*|>\frac{1-3\e}{1+2\e} |W(B, A)\cap E^*|$ for all $\e<1/4$. In particular, $G^*$ contains a matching of size at least $n-2n/d$ whp.
\end{corollary}
\begin{proof}
Set $A_i=A\cap P_i, B_i=B\cap Q_i$, where $G_i=(P_i, Q_i, E_i)$ are the pieces of the decomposition obtained in Section \ref{subsec:decomp}. For each $(A_i, B_i)$ such that $G_i$ is not an isolated vertex we have by Lemma \ref{lm:mincut} and Theorem \ref{thm:ss}
\begin{equation*}
\begin{split}
\left||W_{G_i}(A_i, B_i)\cap E^*|-p|W_{G_i}(A_i, B_i)|\right|<\e p|C_{G_i}(A_i, B_i)|.\\
\end{split}
\end{equation*}
If $G_i$ is an isolated vertex, we have $|W_{G_i}(A_i, B_i)\cap E^*|=p|W_{G_i}(A_i, B_i)|=0$. Since the latter estimate is stronger than the former, we shall not consider the isolated vertices separately in what follows.

Adding these inequalities over all $i$ we get
\begin{equation}
\begin{split}
\sum_{i=1}^k|W_{G_i}(A_i, B_i)\cap E^*|&\geq p \sum_{i=1}^k|W_{G_i}(A_i, B_i)|-\e p\sum_{i=1}^k|C_{G_i}(A_i, B_i)|.
\end{split}
\end{equation}
Denote the set of edges removed  during the decomposition process by $E_r$. Denote $E_1=E_r\cap W(A, B)$ and $E_2=E_r\cap W(B, A)$. Since $|W(A, B)\cap E^*|=\sum_{i=1}^k|W_{G_i}(A_i, B_i)\cap E^*|+|E_1\cap E'|$ and $\sum_{i=1}^k|W_{G_i}(A_i, B_i)|=|W(A, B)|-|E_1|$, this implies
\begin{equation*}
|W(A, B)\cap E^*| \geq p |W(A, B)|-\e p |C(A, B)|-p|E_1|.
\end{equation*}
Likewise, since $W(B, A)=W(P\setminus A, Q\setminus B)$, we have
\begin{equation*}
|W(B, A)\cap E^*| \leq p |W(B, A)|+\e p |C(A, B)|+p|E_2|.
\end{equation*}
Since $|A|\geq |B|+2n/d$, we have $|W(A, B)|\geq |W(B, A)|+2n$, so 
\begin{equation*}
\begin{split}
|W(A, B)\cap E^*| &\geq p |W(A, B)|-\e p |C(A, B)|-p|E_1|\\
&\geq p (|W(B, A)|+2n)-\e p |C(A, B)|-p|E_1|-p|E_2|\\
&\geq |W(B, A)\cap E^*|-2\e p |C(A, B)|+p(2n-|E_r|)\\
&\geq |W(B, A)\cap E^*|-2\e p |C(A, B)|+pn.
\end{split}
\end{equation*}
By similar arguments $|C(A, B)\cap E^*|\geq (1-\e)p(|C(A, B)|-n)$, i.e. $p|C(A, B)|\leq \frac1{1-\e}|C(A, B)\cap E^*|+pn$. Hence, we have
\begin{equation*}
\begin{split}
|W(A, B)\cap E^*| &\geq |W(B, A)\cap E^*|-2\e p |C(A, B)|+pn\\
&\geq |W(B, A)\cap E^*|-\frac{2\e}{1-\e}|C(A, B)\cap E^*|+(1-2\e)pn\\
&\geq |W(B, A)\cap E^*|-\frac{2\e}{1-\e}\left(|W(A, B)\cap E^*|+|W(B, A)\cap E^*|\right)+(1-2\e) pn,
\end{split}
\end{equation*}
which implies
\begin{equation*}
\begin{split}
|W(A, B)\cap E^*| &>\frac{1-3\e}{1+\e}|W(B, A)\cap E^*|\\
\end{split}
\end{equation*}
for $\e<1/4$. This completes the proof.
\end{proof}

\begin{remark} \label{rmk:logd}
The result in corollary \ref{cor:unbalanced-ws} is tight up to an $O(\ln d)$ factor for $d=\Omega(\sqrt{n})$. 
\end{remark}
\begin{proof}
The following construction gives a lower bound of $n-\Omega\left(\frac{n}{d\ln d}\right)$. Denote by $G_{n, d}$ the graph from Theorem 4.1 in \cite{gkk:rbp08} and denote by $G^*_{n, d}$ a graph obtained by sampling edges of $G_{n, d}$ at the rate of $\frac{c\ln n}{d}$ for a constant $c>0$. Define the graph $G$ as $d$ disjoint copies of $G_{2d\ln d, d}$, and denote the sampled graph by $G^*$. Note that by Theorem 4.1 the maximum matching in each copy of $G^*_{2d\ln d, d}$ has size at most $2d \ln d-1$ whp, and since the number of vertices in $G$ is $N=2d^2\ln d$, the maximum matching in $G^*$ has size at most $N-\Omega\left(\frac{N}{d\ln d}\right)$ whp.
\end{proof}

\subsection{Runtime analysis of the Hopcroft-Karp algorithm}
In this section we derive a bound on the runtime of the Hopcroft-Karp algorithm on the subsampled graph obtained in step \textbf{S2} of our algorithm.
The main object of our analysis is the alternating level graph, which we now define. Given a partial matching of a graph $G=(P, Q, E)$, the alternating level graph is defined inductively. Define sets $A_j$ and $B_j$, $j=1,\ldots, L$ as follows. Let $A_0$ be the set of unmatched vertices in $P$ and let $B_0=\emptyset$. Then let $B_{j+1}=\Gamma(A_j)\setminus \left(\bigcup_{i<j} B_{i}\right)$, where  $\Gamma(A)$ is the set of neighbours of vertices in $A\subseteq V(G)$, and let $A_j$ be the set of vertices matched to vertices from $B_j$. The construction terminates when either $B_{j+1}$ contains an unmatched vertex or when $B_{j+1}=\emptyset$, and then we set $L=j$. We use the notation $A^{(j)}=\bigcup_{k\leq j} A_k, B^{(j)}=\bigcup_{k\leq j} B_k$. 
We now give an outline of the Hopcroft-Karp algorithm for convenience of the reader. Given a non-maximum matching, the algorithm starts by constructing the alternating level graph described above and stops when an unmatched vertex is found. Then the algorithm finds a maximal set of vertex-disjoint augmenting paths of length $L$ (this can be done by depth-first search in $O(m)$ time) and performs the augmentations, thus completing one augmentation phase. It can be shown that each augmentation phase increases the length of the shortest augmenting path. Standard analysis of the run-time for general bipartite graphs is based on the observation that once $\sqrt{n}$ augmentations have been performed, the constructed matching necessarily has size at most $\sqrt{n}$ smaller than the maximum matching.  

We denote the graph obtained by sampling edges of $G$ independently with probability $p=\frac{c\ln n}{d}$ for a constant $c>0$ by $G^*$. Note that $G^*$ is obtained from $G$ by \textit{uniform} sampling. We will make the connection to non-uniform sampling in Theorem \ref{thm:runtime}. For $A\subseteq V(G)$ denote the set of edges in the cut $(A, V(G)\setminus A)$ in $G$ by $\delta(A)$ and the set of edges in the same cut in $G^*$ by $\delta^*(A)$.
Similarly, we denote the vertex neighbourhood of $A$ in $G$ by $\Gamma(A)$ and the vertex neighbourhood in $G^*$ by $\Gamma^*(A)$.
 We consider the alternating level graph in $G^*$ and prove that whp for any partial matching of size smaller than $n-2n/d$ for
each $1\leq j\leq L$ either $|B_{j-1}\cup B_j\cup B_{j+1}|=\Omega(d)$ or $B_j$ expands by at least a factor of $\ln n$ in
either forward or backward direction ($|B_{j+1}|\geq (\ln n)|B_j|$ or $|B_{j-1}|\geq (\ln n) |B_j|$). This implies that $L=O\left(\frac{n\ln
    d}{d\ln \ln n}\right)$, thus yielding the same bound on the length of
the shortest augmenting path by virtue of corollary \ref{cor:unbalanced-ws}. The main technical result of this subsection is

\begin{lemma} \label{lm:main} Let the graph $G^*$ be obtained from the
  bipartite $d$-regular graph $G$ on $2n$ vertices by uniform sampling with
  probability $p$. There exist constants $c>0, \e>0$ such that if $p\geq
  \frac{c\ln n}{\e^2 d}$, then whp for any partial matching in $G^*$ of size
  smaller than $n-2n/d$ there exists an augmenting path of length
  $O\left(\frac{n\ln d}{d\ln \ln n}\right)$.
\end{lemma}

The following expansion property of the graph $G^*$ will be used to prove lemma \ref{lm:main}:
\begin{lemma} \label{lm:expansion}
Define $\gamma(t)=(1-\exp(-t))/t$. For all $t>0$ there exists a constant $c>0$ that depends on $t$ and $\e$ such that if $G^*$ is obtained by sampling the edges of $G$ independently with probability $p>\frac{c\ln n}{d}$, then whp for every set $A\subseteq P$, $|A|\leq t/p$ (resp. $B\subseteq Q$, $|B|\leq t/p$) 
\begin{equation*}
\begin{split}
|\Gamma^*(A)|\geq  (1-\e) d p \gamma(t) |A|.
\end{split}
\end{equation*}
\end{lemma}
\begin{proof}
Consider a set $A\subseteq P$, $|A|\leq t/p$. For $b\in \Gamma(A)$ denote the indicator variable corresponding to the event that at least one edge incident on $b$ and going to $A$ is sampled by $X_{b}$, i.e. $X_b=I_{\lbrace b\in \Gamma^*(A)\rbrace}$. Denote the number of edges between $b$ and vertices of $A$ by $k_b$. We have 
\begin{equation*}
\prob[X_b=1]=1-(1-p)^{k_b}\geq 1-\exp(-k_b p)\geq k_b p \gamma(t),
\end{equation*}
since $k_b p\leq t$ and $e^{-x}\leq 1-\gamma(t) x$ for $x\in [0, t]$.

Hence, 
\begin{equation}  \label{xb-expect}
\expect\left[\sum_{b\in B} X_b\right]\geq p\sum_{b\in B} k_b \geq p|\delta(A)| \gamma(t).
\end{equation}

There are at most $n^s$ subsets $A$ of $P$ of size $s$ and $|\delta(A)|=d|A|$ for all $A$, so we obtain using Chernoff bounds and the union bound
\begin{equation*}
\begin{split}
\prob\left[\exists~A\subseteq P: |\Gamma^*(A)|<(1-\e)pd|A|\gamma(t)\right]<\sum_{s=1}^n n^s\exp\left(-\e^2 p d s \gamma(t)\right)\\
=\sum_{s=1}^n \exp\left(s(1-c \gamma(t))\ln n\right)=O(n^{2-c \gamma(t)}),
\end{split}
\end{equation*}
which can be made $O(n^{-r})$ by choosing $c>(2+r)/\gamma(t)$ for any $r>0$.
\end{proof}

\begin{proofof}{Lemma \ref{lm:main}}
First note that since the partial matching is of size strictly less than $n-2n/d$, by Corollary \ref{cor:unbalanced-ws} there exists an augmenting path with
 respect to the partial matching.

In order to upperbound the length of the shortest augmenting path, we will show that for each $j$, at least one of the following is true:  

\begin{enumerate}
\item $|B_j|\geq d/500$;
\item $|B_{j+1}|\geq d/500$;
\item $|B_{j+1}|\geq (\ln n)|B_{j}|$;
\item $|B_{j-1}|\geq d/500$; 
\item $|B_{j-1}|\geq (\ln n)|B_{j}|$.
\end{enumerate}

It then follows that for each $j$  there exists $j'$ such that $|j-j'|\leq 1+\ln_{\ln n} d$ and $|B_{j'}|\geq d/500$. Hence, there cannot be more than $O\left(\frac{n \ln d}{d \ln \ln n}\right)$ levels in the alternating level graph, so there always exists an augmenting path of length $O\left(\frac{n\ln d}{d\ln \ln n}\right)$.

For each $1\leq j\leq L$, where $L$ is the number of levels in the alternating level graph, we classify the edges leaving $B_j$ into three classes:
(1) $E_{F}$ contains edges that go to $U\setminus A^{(j)}$, (2) $E_M$ contains edges that go to $A_j$, and (3) $E_R$ contains edges that go to $A_{j-1}$.
At least one of $E_F, E_M, E_R$ has at least $(1-\e)p d|B_j|/3$ edges by Lemma \ref{lm:expansion}. We now consider each of these possibilities.

\textbf{Case (A):}
First suppose that $E_F$ contains at least
 $(1-\e)p d|B_j|/3$ edges. Note that since the partial matching has size smaller than $n-2n/d$ by assumption, we have that $|A^{(j)}|\geq |B^{(j)}|+2n/d$. Hence, by Corollary \ref{cor:unbalanced-ws} the number of edges going from $A_j$ to $B_{j+1}$ is at least 
\begin{equation*}
\frac{(1-3\e)(1-\e)}{1+\e} pd|A_j|/3.
\end{equation*}

Suppose first that $|A_j|<1/(5p)$. Then by Lemma \ref{lm:expansion} one has that $|\Gamma^*(A_j)|\geq (1-\e)\gamma(1/5) pd|A_j| $. Let $\beta^*=1+\e-\frac{(1-3\e)(1-\e)}{3(1+\e)}$. Observe that since one edge going out of $A_j$ yields at most one neighbor, at most $(1+\e)pd|A_j|-\frac{(1-3\e)(1-\e)}{1+\e} pd |A_j|/3=\beta^*pd |A_j|$ neighbours of vertices of $A_j$ are outside $B_{j+1}$.  Setting $\e=1/15$, we get that $B_{j+1}$ contains at least $((1-\e)\gamma(1/5)-\beta^*)pd |A_j|>0.011 pd |A_j|>(\ln n) |A_j|$ neighbours of $|A_{j}|$, i.e. $|B_{j+1}|\geq (\ln n) |A_j|=(\ln n) |B_j|$ (this corresponds to case 3 above). Now if $|A_j|\geq 1/(5p)$, one can find $A'\subseteq A_j$ such that $|A'|=\lfloor 1/(5p)\rfloor$ and at least $\frac{(1-3\e)(1-\e)}{1+\e} pd|A'|/3$ edges going out of $A'$ go to $B_{j+1}$, which implies by the same argument that $|B_{j+1}|\geq 0.011pd|A'|\geq d/500$ (this corresponds to case 2 above).

\textbf{Case (B):} Suppose that $E_M$ contains at least $(1-\e)pd|B_j|/3$ edges. Then by the same argument as in the previous paragraph (after first weakening our estimate to $\frac{(1-3\e)(1-\e)}{1+\e} pd|B_j|/3$) we have that $|A_{j}|\geq (\ln n) |B_j|$ if $|B_j|\leq 1/(5p)$. This is impossible when $\ln n>1$ since $|A_j|=|B_j|$. Hence, $|A_{j}|\geq d/500$ by same argument as above, and hence $|B_j|\geq d/500$ (this corresponds to case 1 above). 

\textbf{Case (C):} Suppose that $E_R$ contains at least $(1-\e)pd|B_j|/3$ edges.  By the same argument as above we have that either $|B_{j-1}|\geq d/500$ (this corresponds to case 5 above) or $B_{j-1}\geq (\ln n) |B_j|$ (this corresponds to case 4 above).

This completes the proof.
\end{proofof}

We can now prove the main result of this section:
\begin{theorem}\label{thm:runtime}
  Let the graph $G^*$ be obtained from $G$ using steps \textbf{S1} and
  \textbf{S2} in the algorithm of section \ref{sec:algo}. Then step
  \textbf{S3} takes $O\left(\frac{n^2\ln^2 n}{d \ln \ln n}\right)$ time
  whp, giving a time of $O\left(\frac{n^2 \ln^3 n}{d}\right)$ for the entire
  algorithm.
\end{theorem}
\begin{proof}
We analyze the runtime of step \textbf{S3} in two stages: (1) finding a matching of size $n-2n/d$, and (2) extending the matching of size $n-2n/d$ to a perfect matching.

Note that the strength of edges in $G'$ obtained after \textbf{S1} does not exceed $\frac{(1+\e)cn\ln n}{\e^2 d}$, the maximum degree, with high probability, for a constant $c>0$. Hence, the combination of sampling uniformly in \textbf{S1} and non-uniformly in \textbf{S2} dominates sampling each edge with probability $\Omega\left(\frac{\ln n}{d}\right)$, so we write $G''=(P, Q, E^*\cup E^{**})$, where $E^*$ is obtained from $E$ by sampling uniformly with probability $p=\frac{c\ln n}{d}$ for a sufficiently large $c>0$. The constant $c>0$ can be made sufficiently large so that Lemma \ref{lm:main} applies by adjusting the constant in the sampling in steps \textbf{S1} and \textbf{S2}. Denote $G^*=(P, Q, E^*)$ and note that the proof of Lemma \ref{lm:main} only uses lower bounds on the number of edges incident to vertices in a given set, as well as the number of vertex neighbours of a set of vertices. Hence, since all bounds apply to $G^*$, the conclusion of the lemma is valid for $G''$ whp as well, and we conclude that the maximum number of layers in  an alternating level graph, and hence the length of the shortest augmenting path, is $O\left(\frac{n\ln d}{d \ln \ln n}\right)$. As each augmentation phase takes time proportional to the number of edges in the graph, this implies that the first stage takes $O\left(\frac{n^2\ln^2 n}{d \ln \ln n}\right)$.

Finally, note that each augmentation phase increases the size of the matching by at least $1$, and thus $O(n/d)$ augmentation suffice to extend the matching constructed in the first stage to a perfect matching. This takes $O\left(\frac{n^2\ln n}{d}\right)$ time, so the runtime is $O\left(\frac{n^2\ln^2 n}{d \ln \ln n}\right)$ for step \textbf{S3}, and $O\left(\frac{n^2\ln^3 n}{d}\right)$ overall.
\end{proof}

\begin{remark}
Theorem \ref{thm:runtime} as well as lemma \ref{lm:main} can be slightly altered to show that the runtime of the Hopcroft-Karp algorithm on the subsampled graph from \cite{gkk:rbp08} is $O\left(\frac{n^3\ln^2 n}{d^2\ln \ln n}\right)$. This shows that the approach in \cite{gkk:rbp08} yields an $\tilde O(n^{5/3})$ algorithm, which is better than $O(n^{1.75})$ stated in \cite{gkk:rbp08}.
\end{remark}

\begin{theorem}\label{thm:lowerbound}
For any function $d(n)\geq 2\sqrt{n}$ there exists an infinite family of $d(n)$-regular graphs with $2n+o(n)$ vertices such that whp the algorithm in section \ref{sec:algo} performs $\Omega(n/d)$ augmentations in the worst case. 
\end{theorem}
\begin{proof}
 In what follows we omit the dependence of $d$ on $n$ for brevity. Define $H^{(k)}=(U, V, E)$, $0\leq k \leq d$, to be a $(d-k)$-regular bipartite graph with $|U|=|V|=d$.  The graph $G$ consists of $t$ copies of $H^{(k)}$, which we denote by $\left\lbrace H_j \right\rbrace_{j=1}^t$, where $H_j=H^{(t-j+1)}$, and $2t$ vertices $u_1,\ldots, u_t$ and $v_1, \ldots, v_t$. Each of $u_1,\ldots, u_t$ is connected to all $d$ vertices in the $V$-part of $H_1$, and 
 for  $1\leq j\leq t$, the vertex $v_j$  is connected to all vertices in the $U$-part of $H_j$. The remaining connections are established by adding $t-j$ edge-disjoint perfect matchings between the $U$ part of $H_j$ and the $V$ part of $H_{j+1}$ for all $1\leq j<t$.
 
 Set $t=n/d\leq \sqrt{n}/2\leq d/4$. Note that the strength of edges in $H_j$ is at least $d/4$, so whp there exists a perfect matching in subgraph of $H_j$ generated 
 by the sampling steps \textbf{S1} and \textbf{S2}, for $1 \le j \le t$. Suppose that at the first iteration of the Hopcroft-Karp algorithm a perfect matching is found in each $H_j$, thus leaving unmatched the vertices  $u_1,\ldots, u_t$ and $v_1, \ldots, v_t$.
 Then from this point on, the shortest augmenting path for each pair pair $(u_j, v_j)$ has length $j$, and each augmentation phase of the Hopcroft-Karp algorithm
 will increase the size of the matching by $1$. Hence, it takes $t$ augmentations to find a perfect matching. The number of vertices is $2(d+1)t=2n+o(n)$. 
\end{proof}

\section{Perfect Matchings in Doubly Stochastic Matrices}
\label{sec:bvn}
An $n \times n$ matrix $A$ is said to be {\em doubly stochastic} if every element is
non-negative, and every row-sum and every column-sum is 1. The celebrated
Birkhoff-von Neumann theorem says that every doubly stochastic matrix is a
convex combination of permutation matrices ({\em i.e.}, matchings). Surprisingly, the
running time of computing this convex combination (known as a Birkhoff-von
Neumann decomposition) is typically reported as $O(m^2\sqrt{n})$, even though
much better algorithms can be easily obtained using existing techniques or
very simple modifications. We list these running times here since there does
not seem to be any published record\footnote{This list was compiled by
  Bhattacharjee and Goel and is presented here to provide some context rather
  than as original work.}. After listing the running times that can be obtained
using existing techniques, we will show how proportionate uncrossings can be
applied to this problem to obtain a slight improvement.

\begin{enumerate}
\item An $O(m^2)$-time algorithm for finding a Birkhoff-von Neumann
  decomposition can be obtained by finding a perfect matching in the existing
  graph using augmenting paths (in time $O(mn)$), assigning this matching a
  weight which is the weight of the smallest edge in the matching, subtracting
  this weight from every edge in the matching (causing one or more edges to be
  removed from the support of $A$), and continuing the augmenting path
  algorithm without restarting. When a matching is found, if we remove $k$
  edges then we need to find only $k$ augmenting paths (finding each
  augmenting path takes time $O(m)$) to find another matching, which leads to
  a total time of $O(m^2)$.
\item Let $b$ be the maximum number of significant bits in any entry of
  $A$. An $O(mb)$-time algorithm for finding a single perfect matching in the
  support of a doubly stochastic matrix can be easily obtained using the
  technique of Gabow and Kariv~\cite{gk:edge1982}: repeatedly find Euler tours
  in edges where the lowest order bit (say bit $j$) is 1, and then increase
  the weight of all edges going from left to right by $2^{-j}$ and decrease
  the weight of all edges going from right to left by the same amount, where
  the directionality of edges corresponds to an arbitrary orientation of the
  Euler tour; this eliminates bit $j$ while preserving the doubly stochastic
  property and without increasing the support.
\item An $O(mnb)$-time algorithm to compute the Birkhoff-von Neumann
  decomposition can be obtained using the edge coloring algorithm of Gabow and
  Kariv~\cite{gk:edge1982}.
\end{enumerate}

We now show how our techniques lead to an $O(m\ln^3 n + n^{1.5}\ln n)$-time
algorithm for finding a single perfect matching in the support of a doubly
stochastic matrix. In realistic scenarios, this is unlikely to be better than
$(2)$ above, and we present this primarily to illustrate another application
of our proportionate uncrossing technique. First, define a weighted bipartite
graph $G=(P,Q,E)$, where $P=\{u_1, u_2, \ldots u_n\}$ corresponds to rows of
$A$, $Q = \{v_1, v_2, \ldots,v_n\}$ corresponds to columns of $A$, and
$(u_i,v_j)\in E$ iff $A_{i,j} > 0$. Define a weight function $w$ on edges,
with $w(u_i,v_j) = A_{i,j}$. Let $\coll$ be the collection of all pairs
$(A,B), A\subseteq P, B\subseteq Q, |P| > |Q|$. Since $A$ is doubly
stochastic, the collection $\coll$ is $(1/2)$-thick with respect to $(G, w,
E)$. Let $\colltwo$ be a $(1/2)$-uncrossing of $\coll$. Performing a
Bencz\'{u}r-Karger sampling on $G$ will guarantee (with high probability) that at
least one edge is sampled from every witness set in $W(\colltwo)$, and hence
running the Hopcroft-Karp algorithm on the sampled graph will yield a perfect
matching with high probability. The running time of $O(m\ln^3 n + n^{1.5}\ln
n)$ is just the sum of the running times of Bencz\'{u}r-Karger sampling for
weighted graphs~\cite{benczurkarger96} and the Hopcroft-Karp matching
algorithm~\cite{hk:match73}.

\newpage
\subsection*{Acknowledgments:} We would like to thank Rajat Bhattacharjee for many useful discussions on precursors to this work, and Michel Goemans for suggesting an alternate proof mentioned in remark \ref{rm:other-approach}.

\pdfbookmark[1]{\refname}{My\refname} 

\appendix

\section{Proof of Lemma~\ref{lem:relevant-only}}
\label{append:relevant-only}
Consider any $(A,B)$ where $|A| > |B|, A\subseteq P, B\subseteq Q$.  Define
$A_i = P_i \cap A$ and $B_i = Q_i \cap B$. Fix an $i$ such that $|A_i| >
|B_i|$; such an $i$ is guaranteed to exist. By the definition of relevance,
there exists a pair $(X,Y)\in \coll$ such that $X\subseteq A_i$, and $W(X,Y)
\cap E_R \subseteq W(A_i,B_i)\cap E_R$. By the assumption in the theorem,
there exists an edge $(u,v) \in E^* \cap E_R \cap W(X,Y)$. Since $W(X,Y) \cap
E_R \subseteq W(A_i,B_i)\cap E_R$, it follows that $(u,v) \in E^* \cap E_R
\cap W(A_i,B_i)$. This edge is in $G^*$, and goes from $A_i$ to $Q_i \setminus
B_i$, {\em i.e.}, from $A_i$ to $Q_i \setminus (Q_i \cap B)$, and hence, from $A$ to
$Q\setminus B$. Since the only assumption on $(A,B)$ was that $|A| > |B|$, we
can now invoke Hall's theorem to claim that $G^*$ has a perfect
matching.$\qed$

\section{Proof of Theorem~\ref{thm:Benczur-Karger}}
\label{append:bk}
As mentioned before, the proof is along very similar lines to that of the
Bencz\'{u}r-Karger sampling theorem, but does not follow in a black-box fashion
and is presented here for completeness. The proof relies on the following
result due to Karger and Stein~\cite{kargerstein96}:
\begin{lemma}
\label{lem:karger}
Let $H(V,E)$ be an undirected graph on $n$ vertices such that each edge $e$ has an associated
non-negative weight $\tilde{p}_e$. Let $s^*$ be the value of minimum cut in $H$
under the weight function $\tilde{p}_e$. Then for any $\alpha \geq 1$, the number of cuts in $H$
of weight at most $\alpha s^*$ is less than $n^{2\alpha}$.
\end{lemma}
\begin{proofof}{Theorem~\ref{thm:Benczur-Karger}}
  We will choose $c = 5$. The first part of the proof shows that it is
  sufficient to bound a certain expression that involves only cuts. The second
  part then bounds this expression.

  For the first part, let $\mu(X) = \sum_{e\in X} p_e$ denote the expected
  number of edges chosen from $X$ by the sampling process. If a set $X\in
  \edgesets$ contains an edge $e$ with $p_e = 1$, then that edge will
  definitely be chosen, and that set does not contribute to $$ \sum_{ X \in
    \edgesets } \Pr[ {\rm No~edge~in~}X~{\rm is~chosen~in~}H']$$ and can be
  removed from $\edgesets$. Hence, assume without loss of generality that $p_e
  < 1$ for every edge in $\bigcup _{X\in\edgesets}X$. Define $\tilde{\mu}(X) =
  \sum_{e\in X} \left(\frac{c\ln n}{s_e}\right)$. Now for any set $X\in
  \edgesets$,

  $$\Pr[ {\rm No~edge~in~}X~{\rm is~chosen~in~}H'] =
  \prod_{e\in X}(1-p_e) \le \prod_{e\in X}e^{-p_e} \le e^{-\mu(X)},$$

  where
  $$\mu(X) = \frac{c \ln n}{\gamma}\sum_{e\in X} \frac{1}{s_e} >
  (c\ln n) \sum_{e\in f(X)} \frac{1}{s_e} = \tilde{\mu}(f(X)).$$ Since $f$ is
  a one-one function, it is sufficient to provide an upper-bound on
  $\sum_{C\in{\mathcal C}}e^{-\tilde{\mu}(C)}$.

  For the second part, let $\tilde{\mu}_1,
  \tilde{\mu}_2,\ldots,\tilde{\mu}_{2^n-2}$ be a non-decreasing sorted
  sequence corresponding to the multi-set $\{\tilde{\mu}(C): C\in {\cal
    C}\}$. Define $q_i = e^{-\tilde{\mu}_i}$. Consider an arbitrary cut
  $C$. Any edge in $C$ can have strength at most $|C|$, and hence
  $\tilde{\mu}(C) \ge c \ln n$, and therefore, $q_1 \leq n^{-c}$. So the sum
  of $q_i$ for the first $n^2$ cuts in the sequence is bounded by
  $n^{-c+2}$. We now focus on the remaining cuts.  By Lemma~\ref{lem:karger},
  we know that for any $\alpha \geq 1$, we have $\tilde{\mu}_{n^{2\alpha}}
  \geq \alpha \tilde{\mu}_1$.  Hence
  $$ \tilde{\mu}_{k} \geq \frac{\ln k}{2 \ln n} \tilde{\mu}_1,$$
  which in turn implies that $q_k \leq k^{-c/2}$. Thus

  $$  \sum_{ X \in \edgesets } \Pr[ {\rm No~edge~in~}X~{\rm is~chosen~in~}H'] \leq \sum_{C\in{\mathcal C}}e^{-\tilde{\mu}(C)} \leq \sum_{k=1}^{n^2} q_k + \sum_{k > n^2} q_k \leq  n^{-c+2} +   \sum_{k > n^2} k^{-c/2} = O(n^{-c+2}),$$
  giving us the desired result when we choose $c = 5$.
\end{proofof}

\section{Proof of Lemma~\ref{lem:strength}}
\label{append:strength}
Assume by way of contradiction that no such integer $j$ exists for some pair
of multisets $S_1$ and $S_2$.  Let $K$ be the largest integer in $S_1 \cup
S_2$, and let $\alpha_i$ and $\beta_i$ denote the number of occurrences of $i$
in the multisets $S_1$ and $S_2$ respectively.  Then for all $j \geq 1$, we
have $$ \sum_{i=j}^{K} {\frac{\alpha_i}{i}} \le \gamma \left( \sum_{i =j}^{K}
  {\frac{\beta_i}{i}} \right).$$ Summing the above inequality for all $j \in
\{1..K\}$, we get $$\sum_{i=1}^{K} {\alpha_i} \le \gamma \left( \sum_{i
    =1}^{K}{\beta_i} \right),$$ which is a contradiction since $|S_1| > \gamma
|S_2|$ by assumption.$\qed$

\end{document}